\documentclass[11pt]{article}
\usepackage{ifthen}
\usepackage{mdwlist}
\usepackage{amsmath,amssymb,amsfonts,amsthm}
\usepackage{bm}
\usepackage{mathtools}
\usepackage{url}
\usepackage{hyperref}
\usepackage{enumitem}
\usepackage{graphicx}
\usepackage{xspace}
\usepackage{verbatim}
\usepackage{algorithm}
\usepackage{algpseudocode}
\usepackage[margin=1in]{geometry}
\usepackage{color}
\usepackage{thmtools}
\usepackage{thm-restate}

\def\a{\alpha}
\def\b{\beta}

\def\d{\delta}

\def\ve{\varepsilon}

\def\g{\gamma}

\def\k{\kappa}

\def\z{\zeta}
\def\th{\theta}

\def\l{\lambda}

\def\m{\mu}

\def\p{\pi}

\def\r{\rho}

\def\s{\sigma}
\def\S{\Sigma}

\def\Om{\Omega}

\newcommand{\erf}{\mbox{\text erf}}

\newcommand{\prob}[2][]{\text{\bf Pr}\ifthenelse{\not\equal{}{#1}}{_{#1}}{}\!\left[#2\right]}
\newcommand{\expect}[2][]{\text{\bf E}\ifthenelse{\not\equal{}{#1}}{_{#1}}{}\!\left[#2\right]}

\newcommand{\dtv}{d_{\mathrm {TV}}}
\newcommand{\dk}{d_{\mathrm K}}

\newtheorem{theorem}{Theorem}

\newtheorem{proposition}{Proposition}
\newtheorem{observation}{Observation}

\newtheorem{fact}{Fact}
\newtheorem{lemma}{Lemma}

\newtheorem{case}{Case}

\newtheorem{definition}{Definition}

\newcommand{\ignore}[1]{}
\providecommand{\poly}{\operatorname*{poly}}

\newenvironment{prevproof}[2]{\noindent {\em {Proof of {#1}~\ref{#2}:}}}{$\hfill\qed$\vskip \belowdisplayskip}

\newcommand{\bg}[1]{\medskip\noindent{\bf #1}}

\newcommand{\val}{v}

\definecolor{Red}{rgb}{1,0,0}

\newcommand{\oldbound}[1]{{}}

\def\ieee{0}
\title{On the Structure, Covering, and Learning of Poisson Multinomial Distributions}
\author {
Constantinos Daskalakis\thanks{Supported by a Sloan Foundation Fellowship, a Microsoft Research Faculty Fellowship, and NSF Award CCF-0953960 (CAREER) and CCF-1101491.}\\
EECS, MIT \\
\tt{costis@mit.edu}
\and
Gautam Kamath\thanks{Supported by NSF Award CCF-0953960 (CAREER).}\\
EECS, MIT\\
\tt{g@csail.mit.edu}
\and
Christos Tzamos\thanks{Supported by NSF Award CCF-0953960 (CAREER) and a Simons Award for Graduate Students in Theoretical Computer Science.}\\
EECS, MIT\\
\tt{tzamos@mit.edu}
}
\begin{document}
\addtocounter{page}{-1}
\maketitle
\thispagestyle{empty}

\begin{abstract}
An $(n,k)$-{\em Poisson Multinomial Distribution} (PMD) is the distribution of the sum of $n$ independent random vectors supported on the set ${\cal B}_k=\{e_1,\ldots,e_k\}$ of standard basis vectors in~$\mathbb{R}^k$. We prove a structural characterization of these distributions, showing that, for all $\ve >0$, any $(n, k)$-Poisson multinomial random vector is $\ve$-close, in total variation distance, to the sum of a discretized multidimensional Gaussian and an independent $(\text{poly}(k/\ve), k)$-Poisson multinomial random vector. Our structural characterization extends the multi-dimensional CLT of~\cite{ValiantV11}, by simultaneously applying to all approximation requirements~$\ve$. In particular, it overcomes  factors depending on $\log n$ and, importantly, the  minimum eigenvalue of the PMD's covariance matrix.

We use our structural characterization to obtain an $\ve$-cover, in total variation distance, of the set of all $(n, k)$-PMDs, significantly improving the cover size of~\cite{DaskalakisP08,DaskalakisP15}, and obtaining the same qualitative dependence of the cover size on $n$ and $\ve$ as the $k=2$ cover of~\cite{DaskalakisP09,DaskalakisP14}. We further exploit this structure to show that $(n,k)$-PMDs can be learned to within $\ve$ in total variation distance from $\tilde{O}_k(1/\ve^2)$ samples, which is near-optimal in terms of dependence on $\ve$ and independent of $n$.  In particular, our result generalizes the single-dimensional result of~\cite{DaskalakisDS12b} for Poisson binomials to arbitrary dimension. Finally, as a corollary of our results on PMDs, we give a $\tilde{O}_k(1/\ve^2)$ sample algorithm for learning $(n,k)$-sums of independent integer random variables (SIIRVs), which is near-optimal for constant $k$.
\end{abstract}

\newpage

\section{Introduction} \label{sec:intro}

Poisson Multinomial Distributions (PMDs) are one the most basic nonparametric multidimensional families of distributions. They express the distribution of how many out of $n$ thrown balls will fall into $k$ bins, when the balls (perhaps because of weight or other characteristics) have different biases towards falling into the different bins. Mathematically, a $(n,k)$-PMD is the distribution of the sum $\sum_{i=1}^n X_i$ of $n$ independent random vectors $X_i$ supported on the set ${\cal B}_k=\{e_1,\ldots,e_k\}$ of standard basis vectors in $\mathbb{R}^k$. In particular, a $(n, k)$-PMD requires for its description $n \cdot (k-1)$ probabilities, specifying the distribution of each summand random vector.

\smallskip In this paper, we advance our understanding of the structure and learnability of this fundamental family of distributions by studying the following questions:

\begin{enumerate}
\item Can we approximate PMDs via simpler distributions such as multi-dimensional Gaussians or Poissons? Do they always ``behave as'' discretized multi-dimensional Gaussians or Poissons? If not, what is the range of  possible ``behaviors'' that PMDs may exhibit? \label{item:question 1}

\item Given $n$, $k$ and $\ve$, is there a small set of distributions that $\ve$-cover, in total variation distance, the set of all $(n, k)$-PMDs? And, how does the size of the cover scale with $n$, $k$ and $\ve$? \label{item:question 2}

\item How many samples from a $(n,k)$-PMD do we need to  learn its density to within $\ve$ in total variation distance? What is the dependence of the learning complexity on the size $n^{O(k)}$ of their support? \label{item:question 3}
\end{enumerate}

\paragraph{Structure of PMDs} It is hard to do justice to the probability literature studying Question~\ref{item:question 1}. The multi-dimensional CLT informs us that the limiting behavior of $(n,k)$-PMDs, as $n \rightarrow +\infty$, is Gaussian, under conditions on the eigenvalues of the summands' covariance matrices; see, e.g., \cite{VanderVaart00}.\footnote{When we approximate some $(n,k)$-PMD or refer to the eigenvalues of its covariance matrix, we typically project the PMD onto a $(k-1)$-dimensional space, e.g. by excluding one of its coordinates, as otherwise the covariance matrix always has a $0$ eigenvalue and the distribution does not have full-dimensional support.} The CLT is quantified for finite $n$ by the multi-dimensional Berry-Esseen theorem, which bounds the difference between the probability masses assigned to convex (or a bit more general) subsets of $\mathbb{R}^k$ by a $(n,k)$-PMD and the multi-dimensional Gaussian distribution with the same mean vector and covariance matrix, with the bound's quality typically degrading as the PMD's covariance matrix tends to singularity; see, e.g.,~\cite{Bentkus05}. More recently, Valiant and Valiant~\cite{ValiantV11} provide a bound in total variation distance, between a $(n, k)$-PMD and the corresponding {\em discretized} multi-dimensional Gaussian, whose quality degrades mildly with $n$ and worse with the minimum eigenvalue of the PMD's covariance matrix~(see Theorem~\ref{thm:val}).\footnote{Notice that bounds on total variation distance are stronger than bounds on the probabilities of all events defined by convex sets in~$\mathbb{R}^k$ that Berry-Esseen-type theorems establish.} Finally, older results using Stein's method bound the total variation distance between a $(n,k)$-PMD and a multivariate Poisson~\cite{Barbour88,DeheuvelsP88}, or a (bona fide) multinomial distribution~\cite{Loh92}. 

In summary, known bounds show that a $(n,k)$-PMD can be approximated by simpler, $\text{poly}(k)$-parameter, distributions, but the quality of their approximation depends on the first few moments of the PMD or its summands.
 Our goal instead is to provide universal approximation theorems showing how to approximate a given $(n,k)$-PMD by simpler distributions for {\em any desired approximation $\ve$} and {\em without assumptions about the moments of the PMD or its summands}. Our main structural theorem is the following.
 \begin{theorem}[PMD Structure]\label{thm: structural}
 For all $n, k \in \mathbb{N}$, and all $\ve >0$, a $(n, k)$-Poisson multinomial random vector is $\ve$-close, in total variation distance, to the sum of a discretized multidimensional Gaussian and an independent $({{\rm poly}}(k/\ve),k)$-Poisson multinomial random vector.
 \end{theorem}
\noindent By introducing the independent $({\rm poly}(k/\ve),k)$-PMD, our structural result side-steps the degradation of the CLT bound of~\cite{ValiantV11} with $\log n$ and the smallest eigenvalue of the PMD's covariance matrix, correcting it to any desired approximation $\ve$. Interestingly, there may be directions where the variance of the discretized Gaussian used in our result may be arbitrarily far from that of the approximated PMD. The  sparse PMD added to the Gaussian serves to correct the variance in those directions, but does so in a correlated manner across several directions.  Moreover, while \cite{ValiantV11} discretize their approximating multidimensional Gaussian to the closest lattice point, our discretization is more faithful to the structure of its covariance matrix; see Definition~\ref{def:gaussian}.  We provide more intuition about our structural result in Section~\ref{sec:approach}, where we also outline its proof. A more detailed proof of Theorem~\ref{thm: structural} appears in Section~\ref{sec:structurepmd} and a more detailed statement is given as Theorem~\ref{thm:struct}.

%


\paragraph{Covers for PMDs} Building covers for $(n,k)$-PMDs was pursued in~\cite{DaskalakisP08,DaskalakisP15} as a means to develop approximation algorithms for Nash equilibria in anonymous games. These are games where $n$ players share the same action set, say $\{1,\ldots,k\}$, and each player's utility depends on their own choice of action as well as the distribution of how many of the other players choose each of the available actions, but players' utility functions may otherwise be different. It was shown that proper $\ve$-covers, in total variation distance, of $(n,k)$-PMDs\footnote{An $\ve$-cover ${\cal F}_\ve$ of a set of distributions ${\cal F}$ is called {\em proper} iff ${\cal F}_\ve \subseteq {\cal F}$.} imply approximation algorithms for Nash equilibria in these games, whose complexity scales with the size of the cover. Intuitively, this is because switching from a mixed Nash equilibrium to a mixed strategy profile with the same distribution of how many players choose each action does not affect players' payoffs by more than~$\ve$. 

The covers for $(n,k)$-PMDs obtained in the anonymous games papers cited above have size: 
$$n^{O\left(2^{k^2}\cdot \left({f(k) \over \ve} \right)^{6 \cdot k}\right)}\text{, where }f(k) \le 2^{3k-1} k^{k^2+1} k!$$
Such covers are of theoretical interest, their interesting feature being that the size is polynomial in~$n$. Indeed, the standard discretization of the parameters of a PMD's constituent vectors results in covers of size exponential in $n$, so a more delicate ``global'' discretization is needed to obtain covers whose size is polynomial in $n$.

Besides providing an asymptotically smaller search space for Nash equilibria in anonymous games, or any other optimization problem over PMDs,  the polynomial rather than exponential dependence of the cover size on $n$ has direct consequences to the learnability of these distributions; see Theorem~\ref{thm:tournament} (from~\cite{DaskalakisK14}) and~\cite{AcharyaJOS14b} for a similar result, which improve a long line of similar results in the probability literature~\cite{DevroyeL01}. In particular, a cover of polynomial size implies directly that these distributions can be learned from a number of samples logarithmic in $n$, despite their support being polynomial in~$n$. Motivated by such applications of covers to algorithms and learning we use our structural result to obtain an improved cover theorem.

 \begin{theorem}[PMD Covers]\label{thm: cover}
 For all $n, k \in \mathbb{N}$, and $\ve >0$, there exists an $\ve$-cover, in total variation distance, of the set of all $(n, k)$-PMDs whose size is
 $$n^{k^2} \cdot \min\left\{ 2^{{{\rm poly}(k/\ve)}}, 2^{O(k^{5k} \cdot \log^{k+2} (1/\ve))}\right\}.$$
 \end{theorem}
 
We make a few remarks about our cover. First, the cover is non-proper, containing distributions that are of the form specified in Theorem~\ref{thm: structural}, i.e. are convolutions of a discretized Gaussian and a PMD. Moreover, it is straightforward to see that any cover has size at least $n^{\Omega(k)}$ and at least $(1/\ve)^{\Omega(k)}$. For the first lower bound, count the number of $(n,k)$-PMDs whose summands are deterministic. For the second, count the number of $(1,k)$-PMDs whose probabilities are integer multiples of $\ve$. So, for fixed $k$, our bound has the right qualitative dependence on $n$ (namely polynomial), and a near-right dependence on $1/\ve$ (namely quasi-polynomial rather than polynomial). Moreover, it obtains the same qualitative dependence on $n$ and $\ve$ as the $k=2$ cover of~\cite{DaskalakisP09,DaskalakisP14}, namely polynomial in $n$ and quasi-polynomial in $1/\ve$.
 
\paragraph{Learning PMDs} In view of tools for hypothesis selection from a cover (see, i.e., Theorem~\ref{thm:tournament}), our cover theorem directly implies  
that $(n, k)$-PMDs can be learned from $O(k^{5k} \cdot \log n \cdot \log^{k+2}(1/\ve)/\ve^2)$ samples. These are near-optimal in terms of~$\ve$, as $\Omega(k/\ve^2)$ samples are necessary even for learning a $(1,k)$-PMD. We show that the dependence on $n$ can be completely removed from the learner, generalizing the results on Poisson Binomial Distributions~\cite{DaskalakisDS12b}.

\begin{theorem}[PMD Learning]\label{thm:PMD learning}
For all $n, k \in \mathbb{N}$ and $\ve>0$, there is a learning algorithm for $(n,k)$-PMDs with the following properties: Let $X=\sum_{i=1}^n X_i$ be any $(n,k)$-Poisson multinomial random vector. The algorithm uses 

{$$\min\left\{ O(k^{5k} \cdot \log^{k+2}(1/\ve)/\ve^2), {{\rm poly}(k/\ve)}\right\}$$}
samples from $X$, runs in time\footnote{We work in the standard ``word RAM'' model in which basic arithmetic operations on $O(\log n)$-bit integers are assumed to take constant time.} 
{$$\min\left\{2^{O(k^{5k} \cdot \log^{k+2} (1/\ve))},2^{{{\rm poly}(k/\ve)}}\right\},$$}
and with probability at least $9/10$ outputs a (succinct description of a) random vector $\tilde{X}$ such that $\dtv(X, \tilde{X}) \le \ve$.
\end{theorem}

\paragraph{Additional Results: Learning $k$-SIIRVs} 

A $(n,k)$-SIIRV is the sum of $n$ independent (single-dimensional) random variables supported on $\{0,\ldots,k-1\}$. SIIRVs generalize Poisson Binomial distributions, which correspond to the case $k=2$. At the same time, SIIRVs can be viewed as projections of PMDs onto the vector $(0,1,\ldots,k-1)$. In particular, if $X$ is a $(n,k)$-SIIRV, there exists a $(n,k)$-Poisson multinomial random vector $Y$, such that $X=(0,1,\ldots,k-1)^{\rm T} \cdot Y$.

Recent work has established that $(n,k)$-SIIRVs can be learned from ${\rm poly}(k/\ve)$ samples, independent of $n$, when even learning a $(1,k)$-SIIRV already requires $\Omega(k/\ve^2)$ samples~\cite{DaskalakisDOST13}. 
A question arising from this work is finding the optimal dependence of the sample complexity on~$\ve$. 
Demonstrating the expressive power of PMDs, as a corollary of our cover result, we show that the optimal dependence is actually $\tilde{O}_k(1/\ve^2)$.
\begin{theorem}[SIIRV Learning]\label{thm:SIIRV learning}
For all $n, k \in \mathbb{N}$ and $\ve>0$, there is a learning algorithm for $(n,k)$-SIIRVs with the following properties: Let $X=\sum_{i=1}^n X_i$ be any $(n,k)$-SIIRV. The algorithm uses 
{$ k^{5k} \cdot {O}(\log^{k+2}(1/\ve)/\ve^2)$}
samples from $X$, 
runs in time {$2^{O( k^{5k}\cdot \log^{k+2}(1/\ve))}$},
and with probability at least $9/10$ outputs a 
random vector $\tilde{X}$ such that $\dtv(X, \tilde{X}) \le \ve$.
\end{theorem}
Simultaneous work by Diakonikolas, Kane and Stewart \cite{DiakonikolasKS15} takes a direct approach to solving this problem.
Using Fourier-based methods, they give a polynomial-time algorithm which requires $\tilde O(k/\ve^2)$ samples, obtaining near-optimal dependence on both $k$ and $\ve$. 

\subsection{Approach} \label{sec:approach}

\paragraph{Structure} The multi-dimensional nature of PMDs poses challenges in understanding their structure. The projection of a $(n,k)$-Poisson multinomial random vector onto each standard basis vector is a $n$-Poisson Binomial random variable, i.e.~distributed as the sum of $n$ independent indicators. Depending on our choice of $\ve$, the latter may be $\ve$-close (in total variation distance) to a discretized Normal distribution (``heavy projection'') or a distribution whose essential support is a length $O(1/\ve^3)$ subinterval of $\{0,\ldots,n\}$ (``light projection'')~\cite{DaskalakisP14}. Intuitively, one would like to aggregate all heavy projections into a discretized multi-dimensional Gaussian and all light projections into a distribution of small support, independent of $n$. However, projections onto different standard basis vectors may be correlated, and they cannot be disentangled this simply.

In fact, even if all projections of a PMD onto the standard basis vectors are heavy---even if they have variance super-polynomial in $k/\ve$, it is still unclear whether the PMD can always be well approximated by a discretized multi-dimensional Gaussian. In particular, the multi-dimensional CLT of Valiant and Valiant~\cite{ValiantV11} (Theorem~\ref{thm:val}) does pay a penalty that scales with $\log n$.


Finally, projections onto non-standard basis vectors may behave more erratically. As we pointed out earlier, the projection of a $(n,k)$-PMD onto the vector $\vec{v}=(0,1,\ldots,k-1)$ is a $(n,k)$-SIIRV, which need not be log-concave or even unimodal, and could even exhibit ``mod-structure'' and be $n$-modal; think of the distribution of $Y+2 \cdot Z$ where $Z$ is sampled from a Binomial$(n,0.5)$ and $Y$ is a Bernoulli$(1/3)$. Whichever simpler distribution we identify to approximate a given $(n,k)$-PMD thus needs to respect the potential mod-structure that the PMD's projection onto $\vec{v}$, its permutations or other integral vectors may exhibit. 

%
%

Our analysis sidesteps the difficulties identified above by showing that, for all $\ve$, $n$, $k$, a $(n,k)$-Poisson multinomial random vector is  $\ve$-close to the sum of a discretized Gaussian and an independent $({\rm poly}(k/\ve),k)$-Poisson multinomial random vector. Roughly speaking, the Gaussian absorbs  the variance in the heavy dimensions, and explains the correlation between light and heavy dimensions, while the sparse PMD explains the remaining variance in the light dimensions. Of course, what dimensions are ``light'' and ``heavy'' in the above discussion depends on our desired approximation~$\ve$.

At the heart of our proof lies the aforecited CLT by Valiant and Valiant \cite{ValiantV11}, approximating a Poisson Multinomial by a discretized Gaussian. There are several issues with its application here: the accuracy of the approximation cannot be made an arbitrary $\ve$, but worse, it deteriorates (logarithmically) as we increase $n$ or decrease the minimum eigenvalue of the covariance matrix of the PMD. The main intuition behind our structural theorem and the main technical roadblock for its proof lies in avoiding paying these two penalties.

To mitigate the latter cost (corresponding to the smallest eigenvalue), we use a stripped down version of the trickle-down sampling procedure from~\cite{DaskalakisP08} to round the parameters of our given PMD. This allows us to shift the parameters of the PMD's constituent random vectors such that they are either equal to $0$ or $1$, or sufficiently far from $0$ or $1$. A coordinated ``rounding'' of these parameters combined with a coupling argument and single-dimensional Poisson approximations allow us to argue that the effect of the rounding is small in the total variation distance of the resulting PMD compared to the original PMD. Each constituent random vector in the resulting PMD now has decent variance in every axis direction where it has non-zero variance. Partitioning the PMD's constituent vectors into sets based on the axis directions where they have non-zero variance, we get that the minimum eigenvalue of each resulting sub-PMD is large in the span of these directions\ifnum\ieee=0 ; see Proposition~\ref{prop:eigcov}\fi.\footnote{Again, as pointed out earlier, when we refer to the eigenvalues of the covariance matrix of a PMD spanning a certain subspace, we always project the PMD onto a subspace of one dimension less, as otherwise the covariance matrix always has a $0$ eigenvalue since the distribution does not have full-dimensional support.} 
\ifnum\ieee=0
Details about this step are given in Section~\ref{sec:rounding}.
\fi

To avoid paying the logarithmic cost  in the value of $n$ (the number of summands) which appears in the CLT, we repeatedly partition and sort the random vectors into buckets. The sub-PMD corresponding to each bucket will have the property that the logarithm of the number of summands is negligible compared to the minimum eigenvalue of its covariance matrix, so that we can apply the central limit theorem from~\cite{ValiantV11}. We note that there will be a small number of  random vectors which do not fall into a bucket that has this property -- these leftover vectors result in the sparse Poisson Multinomial component in our structural result. 
\ifnum\ieee=0
Details about this step are given in Section~\ref{sec:vvclt}.
\fi

The above approximations result in a distribution comprising several discretized Gaussians and a sparse Poisson multinomial. We subsequently merge all component discretized Gaussians into a single distribution. It is well-known that the sum of two Gaussians is another Gaussian whose parameters are equal to the sum of the parameters of its two components. The same is not true for discretized Gaussians, and we must quantify the error induced by this merging operation.
\ifnum\ieee=0
More details are provided in Section~\ref{sec:merging}.
\fi

Our structural results are described further in Section~\ref{sec:structurepmd}.

\paragraph{Cover} We provide two covers for $(n,k)$-PMDs, which are advantageous for different regimes of $k$ and $\ve$. 
The first cover follows directly from Theorem \ref{thm:struct}, which gives a structural characterization of a PMD as the sum of an appropriately discretized Gaussian and a $(\poly(k/\ve),k)$-PMD.
We simply take an additive grid over all the parameters of this characterization to achieve a cover size which is polynomial in $n$ and exponential in $k$ and $1/\ve$.

Similar to \cite{DaskalakisP14}, we can reduce the dependence of the cover size to pseudo-polynomial in $1/\ve$, albeit at an increased cost in $k$.
This is done using a generalization of the moment matching techniques known for Poisson Binomial distributions.
At a high level, this avoids the naive gridding over all $(\poly(k/\ve),k)$-PMDs by filtering out the ones with unique ``moment profiles,'' which describe the first several moments of the distribution.
We prove that any two distributions with matching moment profiles will have small total variation distance by leveraging results by Roos on Krawtchouk approximations to PMDs \cite{Roos02}.

A further description of our cover results is provided in Section~\ref{sec:coverpmd}.

\paragraph{Learning} Our cover theorem (Theorem~\ref{thm: cover}) directly implies (using Theorem~\ref{thm:tournament}) that $(n,k)$-PMDs can be learned from $O(\log N /\ve^2)$ samples, where $N$ is the size of our cover. Given that $N$ is polynomial in $n$, the resulting sample complexity is logarithmic in $n$. To remove the dependence on $n$ from our sample complexity, we need to exploit not just the size but also the structure of the cover. 

In particular, we know from our structural characterization (Theorem~\ref{thm: structural}) that any $(n,k)$-Poisson Multinomial random vector is $\ve$-close to the sum of a discretized multi-dimensional Gaussian and an independent $({\rm poly}(k/\ve),k)$-PMD. The dependence of the cover size on $n$ is due to enumerating over a cover of discretized multi-dimensional Gaussians, as enumerating over $({\rm poly}(k/\ve),k)$-PMDs has no dependence on $n$. The challenge is this: given sample access to an unknown $(n,k)$-PMD can we zoom in to a smaller set of candidate discretized multi-dimensional Gaussians whose size is independent of $n$ and which suffice for the purposes of guaranteeing an approximation to the unknown PMD?

Let us start with an easier task. Suppose that our structural theorem decides that a $(n,k)$-PMD is $\ve$-close in total variation distance to a discretized multi-dimensional Gaussian. In this case, is it possible to recover the Gaussian from ${\rm poly}(k/\ve)$ samples from the PMD? Intuitively the answer should be ``yes,'' as learning a multi-dimensional Gaussian to within $\ve$ in total variation distance is feasible from  $O(k/\ve^2)$ samples. Only there are two complications. First, we are seeking to actually learn a discretized multi-dimensional Gaussian and, most importantly, we do not have sample access to the Gaussian, but a distribution that is $\ve$-close to it in total variation distance. The first complication becomes an issue when the covariance matrix of the Gaussian has minimum eigenvalue that does not scale with some ${\rm poly}(k/\ve)$, which may very well be the case. The second is more severe as it necessitates robust estimators for the moments of a (discretized) multi-dimensional Gaussian that are resilient to an arbitrary movement of $\ve$ probability mass. We are not aware of such estimators even for a (continuous) multi-dimensional Gaussian.

Despite these apparent issues, even in the simple case we are considering, the saving grace comes from a closer examination of the proof of our structural result. When our structural theorem deems a $(n,k)$-PMD approximable by a discretized multi-dimensional Gaussian, we can argue that the covariance matrices $\Sigma$ of the former and $\Sigma_G$ of the latter are spectrally close,  satisfying $|x^{\rm \tiny T} \Sigma x-x^{\rm \small T} \Sigma_G x| \le \ve \cdot x^{\rm \small T} \Sigma x$, for all $x$. So it suffices to learn the covariance matrix of the PMD to which we have direct sample access, thereby obviating the need for a robust estimator. Learning the covariance matrix of a PMD is feasible from ${\rm poly}(k/\ve)$ samples by bounding the kurtosis of any projection of the PMD  (Lemma~\ref{lem:allestimate}).

The bigger challenge is generalizing the approach to when our structural theorem deems a $(n,k)$-Poisson Multinomial random vector $X$ approximable by the sum of a discretized multi-dimensional Gaussian $G$ and a $({\rm poly}(k/\ve),k)$-Poisson Multinomial random vector $Y$. We can enumerate over the latter, but enumerating over the former is too expensive (i.e. will incur a dependence on $n$). So we have to learn it with sample access to $X$. Unfortunately, our spectral approximation is now much weaker. The covariance matrices $\Sigma$ of $X$ and $\Sigma_G$  of $G$ are now related as follows, for all $x$: $|x^{\rm \tiny T} \Sigma x-x^{\rm \small T} \Sigma_G x| \le \ve \cdot x^{\rm \small T} \Sigma x + {\rm poly}(k/\epsilon)$. Hence, for directions $x$ where the variance $x^{\rm \tiny T} \Sigma x$ of $X$ is small, this approximation is quite loose to just approximate $\Sigma_G$ with $\Sigma$.

Our approach is instead  to use samples from $X$ to get a handle on the spectrum of $\Sigma_G$. As before, by bounding the kurtosis of any projection of the PMD, we can produce an estimate $\hat{\Sigma}$ that approximates $\Sigma$ spectrally: for all $x$, $|x^{\rm \tiny T} \Sigma x-x^{\rm \small T} \hat{\Sigma} x| \le \ve \cdot x^{\rm \small T} \Sigma x$ (Lemma~\ref{lem:allestimate}). Then, using Courant minimax principle through the proof of our structural result, we can argue that the $i$-th eigenvalue $\lambda^G_i$ of $\Sigma_G$ and $\hat{\lambda}_i$ of $\hat{\Sigma}$ are related as follows: $|\lambda^G_i - \hat{\lambda}_i| \le O(\ve) \hat{\lambda}_i +{\rm poly}(k/\epsilon)$. So, using the eigenvalues of our learned $\hat{\Sigma}$, we can produce a small  cover for the eigenvalues of $\Sigma_G$. Unfortunately, the corresponding eigenvectors of $\Sigma_G$ and $\hat{\Sigma}$ need not be as closely related, and it is not clear how to grid over those as the ratio of the smallest to the largest eigenvalue may be polynomial in $n$. We show how to use the knowledge of the eigenvalues and the spectral relation between $\hat{\Sigma}$ and $\Sigma_G$ to produce a small cover over matrices $\hat{\Sigma}_G$ (and not eigenvectors) such that at least one matrix in the cover spectrally approximates our target $\Sigma_G$. The details are provided in 
\ifnum\ieee=0
Section~\ref{sec:Gaussian}.
\else
the appendix of the full version.
\fi 
At this point, we have a small cover over possible distributions $Y$ and a small cover over possible discretized multi-dimensional Gaussians. So we can select among these hypotheses using Theorem~\ref{thm:tournament}. 

Our learning algorithm is described in Section~\ref{sec:learningpmd}.
\nocite{AcharyaD15}

\section{Preliminaries}
\label{sec:prelim}
\subsection{Parameters}
\label{sec:params}
Throughout this paper, we will repeatedly refer to three key parameters, $c = c(\ve,k) = \poly(\ve/k)$, $t = t(\ve,k) = \poly(k/\ve)$, and $\g = O(1)$.
We set
$$c = \left(\frac{\ve^2}{k^5}\right)^{1+\d_c}, \qquad t = \left(\frac{k^{19}}{c\ve^6}\right)^{1 + \d_t}, \qquad \g = 6 + \d_\g,$$
for constants $\d_c, \d_t, \d_\g > 0$.

\subsection{Definitions}
We start by defining several of the distribution classes we will consider. 
First, and most importantly, we start with a formal definition of Poisson Multinomial Distributions.

\begin{definition}
A $k$-\emph{Categorical Random Variable} ($k$-CRV) is a random variable that takes values in $\{e_1,\dots,e_k\}$ where $e_j$ is the $k$-dimensional unit vector along direction $j$.
$\p(i)$ is the probability of observing $e_i$.
\end{definition}

\begin{definition}
An $(n,k)$-\emph{Poisson Multinomial Distribution} ($(n,k)$-PMD) is given by the law of the sum of $n$ independent but not necessarily identical $k$-CRVs.
An $(n,k)$-PMD is parameterized by a nonnegative matrix $\p \in [0,1]^{n \times k}$ each of whose rows sum to $1$ is denoted by $M^\p$, and is defined by the following random process:
for each row $\p(i,\cdot)$ of matrix $\p$ interpret it as a probability distribution over the columns of $\p$ and draw a column index from this distribution. Finally, return a row vector recording the total number of samples falling into each column (the histogram of the samples).
\end{definition}

We note that a sample from an $(n,k)$-PMD is redundant  -- given $k-1$ coordinates of a sample, we can recover the final coordinate by noting that the sum of all $k$ coordinates is $n$.
For instance, while a Binomial distribution is over a support of size $2$, a sample is $1$-dimensional since the frequency of the other coordinate may be inferred given the parameter $n$.
With this inspiration in mind, we define the Generalized Multinomial Distribution, which is the primary object of study in \cite{ValiantV11}. 

\begin{definition}
  A \emph{Truncated $k$-Categorical Random Variable} is a random variable that takes values in $\{0, e_1,\dots,e_{k-1}\}$ where $e_j$ is the $(k-1)$-dimensional unit vector along direction $j$, and $0$ is the $(k-1)$ dimensional zero vector.
  $\r(0)$ is the probability of observing the zero vector, and $\r(i)$ is the probability of observing $e_i$.
\end{definition}

\begin{definition}
  \label{def:GMD}
  An $(n,k)$-\emph{Generalized Multinomial Distribution} ($(n,k)$-GMD) is given by the law of the sum of $n$ independent but not necessarily identical truncated $k$-CRVs.
  A GMD is parameterized by a nonnegative matrix $\r \in [0,1]^{n \times (k-1)}$ each of whose rows sum to at most $1$ is denoted by $G^\r$, and is defined by the following random process:
  for each row $\r(i,\cdot)$ of matrix $\r$ interpret it as a probability distribution over the columns of $\r$ -- including, if $\sum_{j=1}^k \r(i,j) <1$, an ``invisible'' column $0$ -- and draw a column index from this distribution. Finally, return a row vector recording the total number of samples falling into each column (the histogram of the samples).
\end{definition}
  For both $(n,k)$-PMDs and $(n,k)$-GMDs, we will refer to $n$ and $k$ as the \emph{size} and \emph{dimension}, respectively.

We note that a PMD corresponds to a GMD where the ``invisible'' column is the zero vector, and thus the definition of GMDs is more general than that of PMDs.
However, whenever we refer to a GMD in this paper, it will explicitly have a non-zero invisible column.

While we will approximate the Multinomial distribution with Gaussian distributions, it does not make sense to compare discrete distributions with continuous distributions, since the total variation distance is always $1$.
As such, we must discretize the Gaussian distributions.
We will use the notation $\lfloor x \rceil$ to say that $x$ is rounded to the nearest integer (with ties being broken arbitrarily).
If $x$ is a vector, we round each coordinate independently to the nearest integer.

\begin{definition}
  The $k$-dimensional \emph{Discretized Gaussian Distribution} with mean $\m$ and covariance matrix $\S$, denoted $\lfloor\mathcal{N}(\m,\S) \rceil$, is the distribution with support $\mathbb{Z}^k$ obtained by picking a sample according to the $k$-dimensional Gaussian $\mathcal{N}(\m,\S)$, then rounding each coordinate to the nearest integer.
\end{definition}

As seen in the definition of an $(n,k)$-GMD, we have one coordinate which is equal to $n$ minus the sum of the other coordinates.
We define a similar notion for a discretized Gaussian.
However, we go one step further, to take care of when there are several such Gaussians which live in disjoint dimensions.
By this, we mean that given two Gaussians, the set of directions in which they have a non-zero variance are disjoint.
Without loss of generality (because we can simply relabel the dimensions), we assume all of a Gaussian's non-zero variance directions are consecutive, i.e., the covariance matrix is all zeros, except for a single block on the diagonal.
Therefore, when we add the covariance matrices, the result is block diagonal.
The resulting distribution is described in the following definition.

\begin{definition}
  \label{def:gaussian}
  The \emph{structure preserving rounding} of a multidimensional Gaussian Distribution takes as input a multi-dimensional Gaussian $\mathcal{N}(\m,\S)$ with $\S$ in block-diagonal form.
  It chooses one coordinate as a ``pivot'' in each block, samples from the Gaussian ignoring these pivots and rounds each value to the nearest integer.
  Finally, the pivot coordinate of each block is set by taking the difference between the sum of the means and the sum of the values sampled within the block.
\end{definition}

\section{Structure of PMDs}
\label{sec:structurepmd}
In this section, we show a structural result, stating that any $(n,k)$-PMD is close to the sum of an appropriately discretized Gaussian and a $(\poly(k/\ve),k)$-PMD:

\begin{theorem}\label{thm:struct}
  For parameters $c$ and $t$ as described in Section~\ref{sec:params}, every $(n,k)$-Poisson multinomial random vector is $\ve$-close to the sum of a Gaussian with a structure preserving rounding and a $(tk^2, k)$-Poisson multinomial random vector.
  For each block of the Gaussian, the minimum non-zero eigenvalue of $\S_i$ is at least $\frac{tc}{2k^4}$.
\end{theorem}

\noindent There are three main steps in the proof of this theorem.
\begin{itemize}
  \item[Step 1]
First, we replace our $(n,k)$-PMD with one where all parameters are sufficiently far from $0$ and $1$, while still being close to the original in total variation distance.
To motivate this operation, we introduce one of our main tools in our approach, the central limit theorem of Valiant and Valiant \cite{ValiantV11}, which approximates an $(n,k)$-GMD by a discretized multivariate Gaussian.
\begin{restatable}[Theorem 4 from \cite{ValiantV10}]{theorem}{val}
\label{thm:val}
  Given a generalized multinomial distribution $G^\r$, with $k$ dimensions and $n$ rows, let $\m$ denote its mean and $\S$ denote its covariance matrix, then
  $$\dtv\left(G^\r,\lfloor \mathcal{N}(\m,\S)\rceil\right) \leq \frac{k^{4/3}}{\s^{1/3}} \cdot 2.2 \cdot (3.1 + 0.83 \log n)^{2/3}$$
  where $\s^2$ is the minimum eigenvalue of $\S$.
\end{restatable}
We note that this has an error term which depends on the minimum eigenvalue of the covariance matrix of the GMD.
If we perform this rounding procedure and ignore any zero coordinates, then we are given the guarantee that the minimum eigenvalue will be sufficiently large. 

Recall that in Section \ref{sec:params} we have set $c = \poly(\ve/k)$.
This lemma summarizes the result of the rounding procedure:
\begin{lemma}\label{lem:round}
  For any $c \leq \frac1{2k}$, given access to the parameter matrix $\r$ for an $(n,k)$-PMD $M^\r$, we can efficiently construct another $(n,k)$-PMD $M^{\hat \r}$, such that, for all $i,j$, $\hat \r(i,j) \not\in (0,c)$, and 
  $$\dtv\left(M^{\r},M^{\hat{\r}}\right) < O\left(c^{1/2} k^{5/2} \log^{1/2}\left(\frac{1}{ck}\right)\right).$$
\end{lemma}

The procedure starts by fixing two coordinates $i$ and $j$, and considers all CRVs with a parameter in $i$ which is close to $0$, and has maximum parameter in coordinate $j$.
We move some of the weight in this ``heavy'' coordinate either to or from the ``light'' coordinate, while approximately preserving the overall mean vector of the set of CRVs.

The analysis of this process uses a stripped-down version of the ``trickle-down'' process in \cite{DaskalakisP08}. This gives an approximate way to sample from a PMD, resulting in a distribution which is very close in total variation distance. While we postpone technical details to 
\ifnum\ieee=0
Section~\ref{sec:rounding},
\else
the appendix of the full version,
\fi
roughly speaking, it works as follows.
First, take a sample from the PMD but disregard the values for its light coordinate $i$ and heavy coordinate $j$. Instead, sample a new value for coordinate $i$ according to a Poisson distribution with parameter $\m_i$, the mean value for coordinate $i$. Finally, set coordinate $j$ to ensure that all coordinates of the sample sum to $n$. 
As mentioned before, the rounding process approximately preserves the value of $\m_i$, and thus this alternate sampling procedure is closely coupled for the rounded and original PMD.
Thus, by triangle inequality, the rounded and original PMDs are close in total variation distance.


We repeat this rounding procedure for each $i$ and $j$, eventually leading to all parameters either being equal to or far from $0$ and $1$.
A full description and analysis of the rounding procedure are in 
\ifnum\ieee=0
Section \ref{sec:rounding}.
\else
the appendix of the full version.
\fi

\item[Step 2]
Now, we have a ``massaged'' $(n,k)$-PMD $M^{\hat \r}$, with no parameters lying in the intervals $(0,c)$ or $(1-c,1)$.
Next, we will show how to relate the massaged $(n,k)$-Poisson multinomial random vector to a sum of $k$ Gaussians with a structure preserving rounding plus a ``sparse'' $(\poly(k/\ve), k)$-PMD.
The general roadmap is as follows.
We start by partitioning the constituent $k$-CRVs into $k$ sets, $S_1, \dots, S_k$, based on which basis vector we are most likely to observe.
We work seperately for each set $S_i$ by considering the GMD formed by leaving out the coordinate $i$. Our goal is to use the CLT of Theorem~\ref{thm:val} to bound the total variation distance between the corresponding GMD and a discretized Gaussian with the same mean and covariance matrix.
We must be careful when applying Theorem~\ref{thm:val}, since the bound depends on the size of the GMD. Instead of applying the theorem directly, to get a useful bound, we further partition the set $S_i$ into smaller subsets and apply the theorem to each of the resulting subsets. We can then ``merge'' the resulting discretized Gaussians together using the following lemma whose proof is given in 
\ifnum\ieee=0
Section~\ref{sec:sum-discretized}:
\else
the appendix of the full version: 
\fi
\begin{restatable}{lemma}{merging}
\label{lem:merge}
  Let $X_1 \sim \mathcal{N}(\m_1, \S_1)$ and $X_2 \sim \mathcal{N}(\m_2, \S_2)$ be $k$-dimensional Gaussian random variables, and let $\s = \min_j \max_i \s_{i,j}$ where $\s_{i,j}$ is the standard deviation of $X_i$ in the direction parallel to the $j$th coordinate axis.
  Then $$\dtv\left(\lfloor X_1 + X_2 \rceil,\lfloor X_1 \rceil + \lfloor X_2\rceil\right) \leq \frac k{2\s}.$$
\end{restatable}

In more detail, we partition each set $S_i$ into $2^{k-1}$ subsets, grouping together CRVs according to the dimensions they are non-zero in, i.e. set $S_i^{\mathcal{I}}$ contains all CRVs that are non zero in the coordinates given by set $\mathcal{I} \subseteq [k]\setminus\{i\}$. We then group these sets into buckets, where a set is assigned to a bucket depending on its cardinality; 
bucket $B^l$ gets all sets $S_i^{\mathcal{I}}$ with $|S_i^{\mathcal{I}}| \in [l^\g t,(l+1)^\g t)$, with $\g = O(1)$ and $t = \poly(k/\ve)$ as defined in Section~\ref{sec:params}.
This bounds the ratio between the size and the minimum eigenvalue of the covariance of the GMD within every bucket other than $B^0$.
This allows us to apply Theorem~\ref{thm:val} and replace the CRVs within each bucket $B^l$ for $l \ge 1$ with a discretized Gaussian, leaving us with a $(\poly(2^k/\ve),k)$-GMD consisting of all the CRVs of bucket $B^0$.
To reduce the number of remaining CRVs to polynomial in $k$, we show that by removing only $\poly(k/\ve)$ of these CRVs, we can apply Theorem~\ref{thm:val} again to the rest and obtain another discretized Gaussian. In particular, in 
\ifnum\ieee=0
Section~\ref{sec:vvclt},
\else
the appendix of the full version,
\fi
we prove the following lemma:
\begin{restatable}{lemma}{sparsebin}\label{lem:sparsebin}
  Let $G^{\hat \r_k^0}$ be the $(|B^0|,k)$-GMD induced by the truncated CRVs in bucket $B^0$.
  Given $\hat \r_k^0$, we can efficiently compute a partition of $B^0$ into $S$ and $\bar S$, where $|\bar S| \leq kt$.
  Letting $\m_S$ and $\S_S$ be the mean and covariance matrix of the $(|S|,k)$-GMD induced by $S$, and $G^{\hat \r_k^{\bar S}}$ be the $(|\bar S|,k)$-GMD induced by $\bar S$,
  $$\dtv\left(G^{\hat \r_k^0},\lfloor \mathcal{N}(\m_S,\S_S)\rceil \ast G^{\hat \r_k^{\bar S}}\right) \leq \frac{8.646k^{3/2}\log^{2/3}(2^k t)}{t^{1/6}c^{1/6}}.$$
  Furthermore, the minimum non-zero eigenvalue of $\S_S$ is at least $\frac{t c}{k}$.
\end{restatable}
 
After merging together all discretized Gaussians (at most one coming from each bucket $B^l$ for all $l \geq 0$) by iteratively applying Lemma~\ref{lem:merge}, we are able to approximate each original set of CRVs $S_i$ as the sum of a single discretized Gaussian and a $(\poly(k/\ve),k)$-PMD.
Combining the result from each of the sets $S_i$ of the initial partition, we obtain the sum of $k$ discretized Gaussians and a $(\poly(k/\ve),k)$-PMD.
The details of this step are described in 
\ifnum\ieee=0
Section \ref{sec:vvclt}.
\else
the appendix of the full version.
\fi

\item[Step 3]
The final step is to show that the $k$ discretized Gaussians can be merged into a single Gaussian with a structure preserving rounding. We note that we cannot apply Lemma~\ref{lem:merge} here, since each discretized Gaussian has a different pivot coordinate that has been left out. (Recall that by construction, the CRVs in set $S_i$ are approximated by a discretized Gaussian that leaves out coordinate $i$). We thus need a new tool to enable us to merge Gaussians defined in different dimensions.
The main idea is that if two Gaussians with a structure preserving rounding overlap in some dimension, we can use the common dimension as the pivot.
We then add the mean vectors and covariance matrices to merge the distributions.
Iteratively repeating this process will merge all distributions which overlap in some coordinate. This leaves us with one or many discretized Gaussians that lie in completely disjoint coordinates which we can describe as a single Gaussian with a structure preserving rounding (defining blocks according to the coordinates spanned by each Gaussian).
If these were (continuous) Gaussians, the swapping and merging operations would have no cost, but some care is required when dealing with discretized Gaussians.
There are two costs which we must bound here.
First, we must show that swapping the pivot of a PMD is inexpensive, and second, we need to bound the cost of repeatedly merging Gaussians. 

We bound the cost of swapping the pivot by proving the following lemma:
\begin{restatable}[Total Variation Swap Lemma]{lemma}{swaptv}\label{lem:swaptv}
For $\mu \in \mathbb{R}^k$, positive semidefinite $\Sigma \in \mathbb{R}^{k \times k}$, $n \in \mathbb{Z}$, let 
\begin{itemize}
\item $X_i$ be the distribution $\mathcal{N}(\mu_{-i},\Sigma_{-i}) $, where $\mu_{-i} \in \mathbb{R}^{k-1}$ is $\mu$ with the $i$th coordinate removed, and $\Sigma_{-i} \in \mathbb{R}^{(k-1) \times (k-1)}$ is $\Sigma$ with the $i$th row and column removed;
\item $Y_i$ be the distribution in which we draw a sample $(x_1, \dots, x_{k-1}) \sim X_i$ and return $$(\lfloor x_1 \rceil, \dots, \lfloor x_{i-1} \rceil, (n - \sum_{j=1}^{k-1} \lfloor x_j \rceil), \lfloor x_{i} \rceil, \dots, \lfloor x_{k-1} \rceil).$$
\end{itemize}
Then $\dtv(Y_i, Y_j) \leq \frac{k}{2\sigma}$ for any $i,j \in [k]$, where $\sigma^2 = \max(\sigma_{-i}^2, \sigma_{-j}^2)$ and $\sigma_{-i}^2$ is the smallest eigenvalue of $\Sigma_{-i}$.
\end{restatable}
By applying Lemma~\ref{lem:swaptv}, we can make two discretized Gaussians have the same left out coordinate and then merge them using Lemma~\ref{lem:merge} if at least one of them has large variance in every direction.
While each of the $k$ discretized Gaussians starts with this property (for the dimensions in which it is non-deterministic), it is not clear whether this is true after a sequence of pivot swaps and merges. 

In many cases, swapping the pivot decreases the minimum eigenvalue of the distribution's covariance matrix by a factor of $\poly(k)$.
This is acceptable if we only perform a single swap, but naively applying this bound for a sequence of $k$ swaps and merges results in the minimum eigenvalue dropping by a factor of $k^{O(k)}$. We show that such a bad situation cannot occur, no matter how one performs the sequence of swaps and merges, by proving the following lemma:
\begin{restatable}[Variance Swap Lemma]{lemma}{swapvar}\label{lem:swapvar}
Let $\Sigma^{(1)}, \dots, \S^{(m)} \in \mathbb{R}^{k \times k}$ be a sequence of symmetric positive-semidefinite matrices, and define $S^{(i)} = \{j\ |\ e_j^T \S^{(i)} e_j \neq 0 \}$ to be the set of coordinates in which $\S^{(i)}$ is non-zero.
Furthermore, let $\S = \sum_i \S^{(i)}$ and $S = \cup_i S^{(i)}$.
Suppose the following hold for all $i$:
\begin{enumerate}
\item $\S^{(i)}$ has eigenvalue $0$ with corresponding eigenvector $\vec 1$
\item There exists coordinate $j^* \in S^{(i)}$ such that $\S^{(i)}_{S^{(i)} \setminus \{j^*\}}$ has minimum eigenvalue at least $\l$
\item $\left(\cup_{\ell < i} S^{(\ell)}\right) \cap S^{(i)} \neq \emptyset$
\end{enumerate}
Then, for all $j \in S$,  the minimum eigenvalue of $\S_{S \setminus \{j \}}$ is at least $\frac \l {2k^3}$.
\end{restatable}

The details of this step, the proofs of Lemma~\ref{lem:swaptv} and Lemma~\ref{lem:swapvar} as well as the proof of Theorem \ref{thm:struct} are described in 
\ifnum\ieee=0
Section \ref{sec:merging}.
\else
the appendix of the full version.
\fi
\end{itemize}

\section{Covers for PMDs}
\label{sec:coverpmd}
In this section, we describe a pair of covers for $(n,k)$-PMDs.

The first cover follows directly from Theorem \ref{thm:struct}, which gives a structural characterization of a $(n,k)$-Poisson multinomial random vector as the sum of an appropriately discretized Gaussian and an $(tk^2,k)$-Poisson multinomial random vector.
We grid over all possible mean vectors and covariance matrices for the Gaussian component, and all possible parameter values for the $(tk^2,k)$-PMD.
These are covered by sets of size $(n \cdot \poly(k/\ve))^{k^2}$ and $2^{\poly(k/\ve)}$ respectively, resulting in an overall cover of size $n^{k^2} \cdot 2^{\poly(k/\ve)}$.
\begin{lemma}
\label{lem:nonsparsecover}
 For all $n, k \in \mathbb{N}$, and all $\ve >0$, there exists an $\ve$-cover of the set of all $(n, k)$-PMDs whose size is
 $$n^{k^2} \cdot  2^{{\rm poly}(k/\ve)}.$$
\end{lemma}
The proof of this lemma is presented in 
\ifnum\ieee=0
Section~\ref{sec:naivecover}.
\else
the appendix of the full version.
\fi

The second cover further sparsifies the cover for the $(tk^2, k)$-PMD component, by using a multivariate generalization of the moment matching technique described in \cite{DaskalakisP14}.
This reduces the cover size for this component to $2^{O(k^{5k} \log^{k+2} (1/\ve))}$.
In \cite{Roos02}, Roos shows that a PMD can be written as the weighted sum of partial derivatives of a regular multinomial distribution.
He goes on to show that dropping the higher order derivatives in this sum results in a total variation approximation, where the quality of the approximation depends on the parameters of the PMD and the point at which we evaluate the derivatives.
We take advantage of this tool to obtain an $\ve$-approximation, through a careful partitioning of the CRVs and choice of point at which to evaluate the derivatives of the multinomial distributions.
This implies that any two distributions which have matching ``moment profiles'' (which roughly describe the lower order derivatives of the distribution) are $\ve$-close to each other, and thus only one representative element must be kept from each such equivalence class.
The size of the cover follows by a counting argument on the number of moment profiles.

\begin{lemma}
\label{lem:sparsecover}
 For all $n, k \in \mathbb{N}$, and all $\ve >0$, there exists an $\ve$-cover of the set of all $(n, k)$-PMDs whose size is
 $$n^{k^2} \cdot  2^{O(k^{5k} \log^{k+2} (1/\ve))}.$$
\end{lemma}
The proof of this lemma is given in 
\ifnum\ieee=0
Section~\ref{sec:sparsecover}.
\else
the appendix of the full version.
\fi
We note that this cover can be efficiently enumerated over, using a dynamic program similar to that of \cite{DaskalakisP14}.

By combining these two lemmas, we obtain Theorem~\ref{thm: cover}.

\section{Learning PMDs}
\label{sec:learningpmd}

As mentioned before, Theorem~\ref{thm: cover} combined with Theorem~\ref{thm:tournament} below (taken from \cite{DaskalakisK14}) immediately implies that $(n,k)$-PMDs can be learned from $O(\log N /\ve^2)$ samples, where $N$ is the size of our cover.

\begin{restatable}[Theorem 19 of \cite{DaskalakisK14}]{theorem}{tournament} \label{thm:tournament}
There is an algorithm {\tt FastTournament}$(X, {\cal H},\ve,\delta)$, which is given sample access to some distribution $X$ and a collection of distributions ${\cal H}=\{H_1,\ldots,H_N\}$ over some set ${\cal D}$, access to a PDF comparator for every pair of distributions $H_i, H_j \in {\cal H}$, an accuracy parameter $\ve >0$, and a confidence parameter $\delta >0$.  The algorithm makes
{$O\left({\log {1/ \delta} \over \ve^2} \cdot \log N\right)$} draws from each of $X, H_1,\ldots,H_N$ and returns some $H \in {\cal H}$ or declares ``failure.''  If there is some $H^* \in {\cal H}$ such that $\dtv(H^*,X) \leq \ve$ then with probability at least $1-\delta$ the distribution $H$ that {\tt
FastTournament} returns satisfies $\dtv(H,X) \leq {512} \ve.$ The total number of operations of the algorithm is {$O\left( {\log{1 / \delta} \over \ve^2} \left(N \log N + \log^2 {1 \over \delta}\right) \right)$}.
Furthermore, the expected number of operations of the algorithm is {$O\left( {N\log{N /\d} \over \ve^2}\right)$}.
\end{restatable}

Theorem~\ref{thm:tournament} is using a tournament-style algorithm for hypothesis selection, which takes a set of candidate distributions and outputs one which is $O(\ve)$-close to the unknown distribution (if such a distribution exists)\footnote{We note that this tournament additionally requires a ``PDF comparator,'' which we describe for our setting in
\ifnum\ieee=0
Section~\ref{sec:pmdpdf}.
\else
the appendix of the full version.
\fi
}.
Given that $N$ is polynomial in $n$, the resulting sample complexity is logarithmic in $n$. 
To remove the dependence on $n$ from our sample complexity, we need to exploit not just the size but also the Gaussian structure of the cover. 
Instead of trying all possible Gaussians that the cover could describe, we instead estimate the moments of the Gaussian directly.

\ifnum\ieee=0
\def\figwidth{16.2cm}
\else
\def\figwidth{17.4cm}
\fi
\begin{figure}
\framebox{
\begin{minipage}{\figwidth}
\begin{enumerate}
\item Guess the block structure/partition of the coordinates.
\item Estimate (using a single sample) the number of CRVs in each block.
\item For each Gaussian in the block structure, use $\poly(k)/\ve^2$ samples to find its mean vector and covariance matrix, as follows:
\begin{enumerate}
\item With $\poly(k)/\ve^2$ samples, estimate the mean vector and covariance matrix of the PMD.
\item Convert these estimates to the mean and covariance of the Gaussian by searching over a spectral cover of positive semidefinite matrices.
\end{enumerate}
\item Guess the sparse component by enumerating over elements in either of the two covers.
\item Run a tournament on the set of guessed distributions to identify one which is $\ve$-close.
\end{enumerate}
\end{minipage}}
\ifnum\ieee=0
\vspace{-10pt}
\else
\fi
\caption{Steps of the learning algorithm}
\end{figure}

Our strategy will not be to generate an $\ve$-cover for \emph{all} $(n,k)$-PMDs, 
but instead we take samples and select only distributions from our cover which are consistent with the data.
Similar to before, we will apply Theorem~\ref{thm:tournament} to do hypothesis selection but instead of applying it to the complete cover resulting from Theorem~\ref{thm: cover}, we will apply it to a much smaller set of hypothesis that we obtain after making several ``guesses'' for the parameters of our distribution.
At least one set of these parameters will be sufficiently accurate to obtain an $\ve$ total variation distance guarantee and we will be able to determine a good candidate using Theorem~\ref{thm:tournament}.

The first step of our learning algorithm is to guess the block-diagonal structure of the Gaussian component of our distribution by guessing the partition of the coordinates and choosing an arbitrary pivot within each block.
This requires at most $k^k$ guesses.
Note that any choice of pivot in the partition is acceptable (as shown in Lemma \ref{lem:swaptv} above).

The next step is to guess the sum of the means for the Gaussian component within each block. We need this to know how to fill in the pivot coordinate once we sampled the rest of the coordinates in the block.
This will be the number of CRVs which result in this block of the Gaussian component, and thus an integer between $0$ and $n$.
Since the total variation distance between the sampled distribution and the distribution from the cover is at most $\ve$, with probability at least $1 - \ve$, the sample has non-zero probability to be generated by the distribution from the cover.
In this case, the sum of the sample's values within each block will be equal to the sum of the means from the Gaussian component, plus the contribution from the sparse $(tk^2,k)$-PMD component.
Therefore, for each block, we can guess the sum of the means via the following procedure:
Take a single sample $X \in \mathbb{R}^k$, and for each block $\mathcal{B}$, guess the sum of the means to be $\sum_{i \in \mathcal{B}} X_i - \ell$, for all $\ell \in \{0,1,\dots,tk^2\}$.
Since there are at most $k$ blocks, this requires $(tk^2 + 1)^k$ guesses.

Next, we estimate the mean and covariance of the Gaussian component for each block. 
We need to estimate them accurately enough in order to learn each block of the discretized Gaussians to within $O(\ve/k)$ in total variation distance. A useful tool for showing this is the following proposition:

\begin{restatable}{proposition}{tvtodirectional}
  \label{lem:tv-to-directional}
  Let $\m,\m' \in \mathbb{R}^k$ and $\S, \S' \in \mathbb{R}^{k \times k}$, such that for all $y \in \mathbb{R}^k$
  $$|y^T(\mu' - \mu)| \leq \ve \sqrt{y^T \S y} ~~ \text{and} ~~ |y^T(\S' - \S)y| \leq \ve y^T \S y.$$
  Then 
$$\dtv(\mathcal{N}(\m, \S),\mathcal{N}(\m', \S')) \leq 2\ve k.$$
\end{restatable}

\noindent Proposition~\ref{lem:tv-to-directional} implies that, in order to achieve the required bound in total variation distance, it suffices to get an estimate that approximately matches the mean and variance of the Gaussian component in every direction.  In 
\ifnum\ieee=0
Section~\ref{sec:momentestimate},
\else
the appendix of the full version,
\fi
we prove Lemma~\ref{lem:allestimate} which shows that using $\poly(k)/\ve^2$ samples from the PMD, we can get an estimate of the mean and covariance matrix that achieves this guarantee in every direction.
However, this estimate is with respect to the PMD we are sampling from and \emph{not} with respect to the Gaussian component, which is the guarantee we desire.

\begin{restatable}{lemma}{allestimate}
  \label{lem:allestimate}
  Given sample access to a $(n,k)$-PMD $X$ with mean $\m$ and covariance matrix $\S$ (with minimum eigenvalue at least $1$), there exists an algorithm which can produce estimates $\hat \mu$ and $\hat \S$ such that with probability at least $9/10$:
  $$|y^T(\hat \mu - \mu)| \leq \ve \sqrt{y^T \S y} ~~ \text{and} ~~ |y^T(\hat \S - \S)y| \leq \ve y^T \S y $$
  for all vectors $y$.

  The sample and time complexity are $O(k^4/\ve^2)$.
\end{restatable}

In order to obtain a guarantee for the Gaussian component, we observe that there are two possible sources of errors in our estimation:
\begin{itemize}
  \item The first source of error comes from the rounding step. In proving our structural result, the real PMD had to be rounded so that no CRV has any probability that is in the range $(0,c)$, which affected the mean and covariance. 
  In 
  \ifnum\ieee=0
  Section~\ref{sec:roundingmoments}, 
  \else
  the appendix of the full version,
  \fi
  we show that this only affects the mean and variance in each direction up to a small multiplicative factor. 
  \item The second source of error is due to the existence of the sparse component creates an additional additive error in each direction. This error might be very significant in some directions as the variance of the Gaussian component can be very small compared to the number of sparse CRVs.
\end{itemize}
Understanding that our estimation is off by an additive error and a multiplicative error, we show how to efficiently correct this estimation by searching around it for the underlying covariance matrice of the Gaussian distribution. In particular, we obtain a cover of positive semidefinite matrices that are close to the estimated covariance matrix and which contains a good approximation to the covariance matrix of the underlying Gaussian.
This is challenging because the above two sources of error might affect the spectrum of the covariance matrix significantly. However, we are able to tackle this issue by carefully guessing appropriate corrections to the eigenvectors and eigenvalues of the matrix.
We prove Lemma~\ref{lem:psdcover} which states that this cover has cardinality at most $(k/\ve)^{O(k^2)}$, and thus we can get a very accurate estimate for the underlying Gaussian distribution by guessing different points in the cover.

\begin{restatable}{lemma}{psdcover}
  \label{lem:psdcover}
  Let $A$ be a symmetric $k \times k$ PSD matrix with minimum eigenvalue $1$ and let $S$ be the set of all matrices $B$ such that $|y^T(A - B)y| \le \ve_1 y^T A y + \ve_2 y^T y$ for all vectors $y$, where $\ve_1 \in [0,1/4)$ and $\ve_2 \in [0,\infty)$. Then, there exists an $\ve$-cover $S_{\ve}$ of $S$ that has size $|S_{\ve}| \le \left(\frac {k (1 + \ve_2)} { \ve} \right)^{O(k^2)}$.
\end{restatable}

At this point, we have a collection of distributions such that at least one is close to the Gaussian component.
We do the same for the sparse PMD component by simply enumerating over all the elements in the cover.
By reading the corresponding term from the statement of Theorem \ref{thm: cover}, this requires $\min\{2^{\poly(k/\ve)}, 2^{O(k^{5k} \cdot \log^{k+2} (1/\ve))}\}$ guesses.


In conclusion, using $\poly(k)/\ve^2$ samples, we have generated a set $\mathcal{S}$ of size 
$$\left(k/\ve\right)^{O(k^2)} \cdot \min\{2^{\poly(k/\ve)}, 2^{O(k^{5k} \cdot \log^{k+2} (1/\ve))}\}$$
 which contains a distribution which is $\ve$-close to the true distribution with constant probability. 
In order to choose a ``good'' distribution from this set, we apply the hypothesis selection algorithm of Theorem \ref{thm:tournament} to obtain a distribution which is $O(\ve)$-close to the unknown distribution with constant probability, which concludes the proof of Theorem \ref{thm:PMD learning}. More details about the learning steps and complete proofs can be found in 
\ifnum\ieee=0
Section~\ref{app:learning}.
\else
the full version.
\fi

\section{Learning $k$-SIIRVs}
\label{sec:learningsiirv}
We demonstrate the expressive power of PMDs by demonstrating their applicability to learning $(n,k)$-SIIRVs.
In particular, we leverage our cover results to give a $ \tilde O_k(1/\ve^2)$ sample algorithm for this problem.

The proof uses the structural result of \cite{DaskalakisDOST13}, which says that any $(n,k)$-SIIRV is close to either a low variance distribution with limited support, or a high variance distribution which enjoys certain Gaussian structural properties.

\begin{restatable}[Corollary 4.8 of \cite{DaskalakisDOST13}]{lemma}{siirvstruct}\label{lem:siirvstruct}
  Let $S = X_1 + \dots + X_n$ be a $(n,k)$-SIIRV for some positive integer $k$.
  Let $\m$ and $\s^2$ be respectively the mean and variance of $S$.
  Then for all $\ve > 0$, the distribution of $S$ is $O(\ve)$-close in total variation distance to one of the following:
  \begin{enumerate}
    \item a random variable supported on $\frac{k^9}{\ve^4}$ consecutive integers with variance $\s^2 \leq 15(k^{18}/\ve^6)\log^2(1/\ve)$; or
    \item the sum of two independent random variables $S_1 + cS_2$, where $c$ is some positive integer $1 \leq c \leq k-1$, $S_2$ is distributed according to $\lfloor \mathcal{N}(\m, \s^2) \rceil$, and $S_1$ is a $c$-IRV; in this case, $\s^2 = \Omega\left(\frac{k^{18}}{\ve^6}\log^2(1/\ve)\right)$.
  \end{enumerate}
\end{restatable}

As we did for PMDs, we will use the tournament based approach, in which we generate a set of probability distributions $\mathcal{S}$, containing at least one distribution which is $\ve$-close to $S$.
We then use Theorem~\ref{thm:tournament} to select a distribution which is $O(\ve)$-close to $S$, using $\tilde O(|\mathcal{S}|/\ve^2)$ samples. 

To cover the former case, we use the PMD cover of Theorem \ref{thm: cover}.
In this setting, the SIIRV has a variance upper bounded by $\poly(k/\ve)$.
By applying a rounding procedure, it can be shown that this can be approximated by an offset $(\poly(k/\ve),k)$-SIIRV.
Recalling that any $(n,k)$-SIIRV can be expressed as the projection of an $(n,k)$-PMD onto the vector $(0,1, \dots, k-1)$ and applying our quasi-polynomial cover result in Theorem \ref{thm: cover} covers this case with $2^{O( k^{5k}\cdot \log^{k+2}(1/\ve))}$ candidates.

To cover the latter case, we first perform $k-1$ guesses for the value of $c \in [k-1]$.
For each guess, we learn the two distributions $S_1$ and $S_2$ separately.
To learn $S_1$, we use the same approach as \cite{DaskalakisDOST13}, which uses the empirical distribution obtained after mapping the samples onto $\{0,1,\dots, c-1\}$ using their residue mod $c$.
Our method for learning $S_2$ is novel -- we first round the value of each sample down to the next multiple of $c$, and examine the distribution on this support, which will be close in total variation distance to $S_2$.
We estimate the moments of this distribution using robust statistical tools, as in \cite{DaskalakisDOST13}.
The empirical median is used to estimate the mean, and a rescaling of the interquartile range is used to estimate the standard deviation.
Thus, we cover this case using only $k-1$ candidates, one for each guess of $c$.

Full details are provided in 
\ifnum\ieee=0
Section~\ref{sec:appendixsiirv}. 
\else
the appendix of the full version.
\fi

\bibliographystyle{alpha}
\bibliography{biblio}
\appendix
\section{Useful Tools}
\subsection{Probability Metrics}
To compare probability distributions, we will require the total variation and Kolmogorov distances:
\begin{definition}
The \emph{total variation distance} between two probability measures $P$ and $Q$ on a $\s$-algebra $F$ is defined by
$$\dtv(P,Q) = \sup_{A \in F} |P(A) - Q(A)| = \frac12\|P - Q\|_1.$$
\end{definition}

Unless explicitly stated otherwise, in this paper, when two distributions are said to be $\ve$-close, we mean in total variation distance.

\begin{definition}
  The \emph{Kolmogorov distance} between two probability measures $P$ and $Q$ with CDFs $F_P$ and $F_Q$ is defined by
  $$\dk(P,Q) = \sup_{x \in \mathbb{R}} |F_P(x) - F_Q(x)|.$$
\end{definition}

We note that Kolmogorov distance is, in general, weaker than total variation distance. 
In particular, total variation distance between two distributions is lower bounded by the Kolmogorov distance.
\begin{fact}
\label{fct:drelation}
$\dk(P,Q) \leq \dtv(P,Q)$
\end{fact}

\subsection{Probabilistic Tools}

We will use the following form of Chernoff/Hoeffding bounds:
\begin{lemma}[Chernoff/Hoeffding]\label{lem:cher}
  Let $Z_1, \dots, Z_m$ be independent random variables with $Z_i \in [0,1]$ for all $i$. Then, if $Z = \sum_{i=1}^n Z_i$ and $\gamma \in (0,1)$,
  $$\Pr[|Z - E[Z]| \geq \g E[Z]] \leq 2\exp(-\g^2 E[Z]/3).$$
\end{lemma}

We note the Dvoretzky-Kiefer-Wolfowitz (DKW) inequality, which is a powerful tool, giving a generic algorithm for learning any distribution with respect to the Kolmogorov metric \cite{DvoretzkyKW56}.
\begin{lemma}{(\cite{DvoretzkyKW56},\cite{Massart90})}
  \label{lem:dkw}
  Suppose we have $n$ IID samples $X_1, \dots X_n$ from a probability distribution with CDF $F$.
  Let $F_n(x) = \frac{1}{n}\sum_{i=1}^n \mathbf{1}_{\{X_i \leq x\}}$ be the empirical CDF.
  Then $\Pr[\dk(F,F_n) \geq \ve] \leq 2e^{-2n\ve^2}$.
  In particular, if $n = \Omega((1/\ve^2) \cdot \log(1/\d))$, then $\Pr[\dk(F,F_n) \geq \ve] \leq \d$.
\end{lemma}

We will use the Data Processing Inequality for total variation distance (see part (iv) of Lemma 2 of \cite{Reyzin11} for the proof).
This lemma says that taking any function of two random variables can only reduce their total variation distance.
Our statement of the inequality is taken from \cite{DaskalakisDOST13}.
\begin{lemma}[Data Processing Inequality for Total Variation Distance]\label{lem:DPI}
  Let $X, X'$ be two random variables over a domain $\Omega$.
  Fix any (possibly randomized) function $F$ on $\Omega$ (which may be viewed as a distribution over deterministic functions on $\Omega$) and let $F(X)$ be the random variable such that a draw from $F(X)$ is obtained by drawing independently $x$ from $X$ and $f$ from $F$ and then outputting $f(x)$ (likewise for $F(X')$).
  Then we have
  $$\dtv\left(F(X),F(X')\right) \leq \dtv\left(X,X'\right).$$
\end{lemma}

Finally, we require a hypothesis selection algorithm.
Roughly, given a set of $N$ distributions with the guarantee that at least one is $\ve$-close to an unknown distribution $X$, we can choose a hypothesis which is $O(\ve)$-close to $X$.
The running time is near-linear in $N$ and the number of samples is logarithmic in $N$.

\begin{definition}
\label{def:pdfcomp}
Let $H_1$ and $H_2$ be probability distributions over some set ${\cal D}$. A {\em PDF comparator for $H_1, H_2$} is an oracle that takes as input some $x \in {\cal D}$ and outputs $1$ if $H_1(x)>H_2(x)$, and $0$ otherwise. 
\end{definition}

\tournament*

\subsection{Bounds for Distances Between Distributions}
\begin{proposition}[Proposition B.4 of \cite{DaskalakisDOST13}]
    \label{prop:ddodtv}
      Let $\m_1, \m_2 \in \mathbb{R}$ and $0 \leq \s_1 \leq \s_2$.
        Then $$\dtv(\mathcal{N}(\m_1,\s_1^2),\mathcal{N}(\m_2, \s_2^2)) \leq \frac{1}{2}\left(\frac{|\m_1 - \m_2|}{\s_1} + \frac{\s_2^2 - \s_1^2}{\s_1^2}\right).$$
\end{proposition}

  \begin{proposition}[Proposition 32 in \cite{ValiantV10}]\label{prop:gaussapprox}
    Given two $k$-dimensional Gaussians $\mathcal{N}_1 = \mathcal{N}(\m_1,\S_1), \mathcal{N}_2 = \mathcal{N}(\m_2,\S_2)$ such that for all $i,j \in [k]$, $|\S_1(i,j) - \S_2(i,j)| \leq \a$, and the minimum eigenvalue of $\S_1$ is at least $\s^2$,
    $$\dtv\left(\mathcal{N}_1, \mathcal{N}_2\right) \leq \frac{\| \m_1 - \m_2\|_2}{\sqrt{2\p\s^2}} + \frac{k\a}{\sqrt{2\p e}(\s^2 - \a)}.$$
  \end{proposition}

\tvtodirectional*
\begin{proof}
  Without loss of generality, assume $\m = 0$, $\S = I$, and $\S'$ is diagonal.
  This can be done by setting $y = Q\Lambda^{-1/2}Q'^Tx$, where $\S = Q \Lambda Q^T$ and $\S' = Q' \Lambda' Q'^T$ are the eigendecompositions of $\S$ and $\S'$. 
  
  This implies that we now have the following guarantees for all $i \in [k]$:
  $$|\mu'_i| \leq \ve ~~ \text{and} ~~ |\S'_{i,i} - 1| \leq \ve.$$

  Since each coordinate is independent and noting that $\S'_{i,i} \geq 1 - \ve$, we can apply Proposition \ref{prop:ddodtv} to each coordinate direction to obtain a total variation distance of $2\ve k$.
\end{proof}

\begin{proposition}[Berry-Esseen theorem \cite{Berry41, Esseen42, Shevtsova10}]
\label{prop:berryesseen}
Let $X_1, \dots, X_n$ be independent random variables, with $E[X_i] = 0, E[X_i^2] = \s_i^2 > 0, E[|X_i|^3] = \r_i < \infty$, and define $X = \sum_{i=1}^n X_i, \s^2 = \sum_{i=1}^n \s_i^2, \r = \sum_{i=1}^n \r_i$.
Then for an absolute constant $C_0 \leq 0.56$,
$$\dk(X, \mathcal{N}(0,\s^2)) \leq \frac{C_0 \r}{\s^{3}}.$$
\end{proposition}

\subsection{Covariance Matrices of Truncated Categorical Random Variables}

First, recall the definition of a symmetric diagonally dominant matrix.
\begin{definition}
  A matrix $A$ is \emph{symmetric diagonally dominant} (SDD) if $A^T = A$ and $A_{ii} \geq \sum_{j\neq i} |A_{ij}|$ for all $i$.
\end{definition}

As a tool, we will use this corollary of the Gershgorin Circle Theorem \cite{Gershgorin31} which follows since all eigenvalues of a symmetric matrix are real.

\begin{proposition}\label{prop:gersh}
  Given an SDD matrix $A$ with positive diagonal entries, the minimum eigenvalue of $A$ is at least $\min_i A_{ii} - \sum_{j \neq i} |A_{ij}|$.
\end{proposition}

\begin{proposition}\label{prop:eigcov}
  The minimum eigenvalue of the covariance matrix $\S$ of a truncated CRV is at least $\r(0) \min_i \r(i)$.
\end{proposition}
\begin{proof}
  The entries of the covariance matrix are
  \begin{align*}
    \S_{ij} &= E[x_ix_j] - E[x_i]E[x_j] \\
            &= \begin{cases}
\r(i) - \r(i)^2 &\mbox{if } i = j \\
    -\r(i)\r(j) &\mbox{else}
               \end{cases}
  \end{align*}
  We note that $\S$ is SDD, since $\sum_{j \neq i} |\S_{ij}| = \r(i) \sum_{j \neq i} \r(j) = \r(i) (1 - \r(i) - \r(0)) \leq \r(i)(1 - \r(i)) = \S_{ii}$.
  Thus, applying Proposition \ref{prop:gersh}, we see that the minimum eigenvalue of $\S$ is at least $\min_i \r(i) (1 - \r(i)) - \r(i) (1 - \r(i) - \r(0)) = \r(0) \min_i \r(i)$.
\end{proof}

\subsection{Sums of Discretized Gaussians}\label{sec:sum-discretized}
In this section, we will obtain total variation distance bounds on merging the sum of discretized Gaussians.
It is well known that the sum of multiple Gaussians has the same distribution as a single Gaussian with parameters equal to the sum of the components' parameters.
However, this is not true if we are summing discretized Gaussians -- we quantify the amount we lose by replacing the distribution with a single Gaussian, and then discretizing afterwards.

As a tool, we will use the following result from \cite{DaskalakisDOST13}:
\begin{proposition}[Proposition B.5 in \cite{DaskalakisDOST13}]\label{prop:DDOST}  
  Let $X \sim \mathcal{N}(\m, \s^2)$ and $\l \in \mathbb{R}$. Then $$\dtv\left(\lfloor X + \l \rceil,\lfloor X \rceil + \lfloor \l \rceil\right) \leq \frac1{2\s}.$$
\end{proposition}
From this, we can obtain the following:
\begin{proposition}\label{prop:DDOSTplus}
  Let $X_1 \sim \mathcal{N}(\m_1, \s_1^2)$ and $X_2 \sim \mathcal{N}(\m_2, \s_2^2)$. 
  Then $$\dtv\left(\lfloor X_1 + X_2 \rceil,\lfloor X_1 \rceil + \lfloor X_2\rceil\right) \leq \frac1{2\s},$$ where $\s = \max_i \s_i$.
\end{proposition}

\begin{proof}
  First, suppose without loss of generality that $\s_1 \geq \s_2$.
  \begin{align*}
    & \dtv\left(\lfloor X_1 + X_2 \rceil,\lfloor X_1 \rceil + \lfloor X_2\rceil\right) \\
 &= \frac12 \sum_{i = -\infty}^\infty \left| \Pr(\lfloor X_1 + X_2 \rceil = i) -  \Pr(\lfloor X_1 \rceil +  \lfloor X_2 \rceil = i) \right| \\
 &= \frac12 \sum_{i = -\infty}^\infty \left| \int_{-\infty}^\infty f_{X_2}(\l) \Pr(\lfloor X_1 + \l \rceil = i)\, d\l -  \int_{-\infty}^\infty f_{X_2}(\l) \Pr(\lfloor X_1 \rceil +  \lfloor \l \rceil = i)\, d\l \right| \\
 &= \frac12 \sum_{i = -\infty}^\infty \left| \int_{-\infty}^\infty f_{X_2}(\l) (\Pr(\lfloor X_1 + \l \rceil = i) - \Pr(\lfloor X_1 \rceil +  \lfloor \l \rceil = i))\, d\l \right| \\
 &\leq \frac12 \sum_{i = -\infty}^\infty  \int_{-\infty}^\infty f_{X_2}(\l) \left|(\Pr(\lfloor X_1 + \l \rceil = i) - \Pr(\lfloor X_1 \rceil +  \lfloor \l \rceil = i))\right|\, d\l  \\
 &= \frac12 \int_{-\infty}^\infty f_{X_2}(\l) \left(\sum_{i = -\infty}^\infty  \left|(\Pr(\lfloor X_1 + \l \rceil = i) - \Pr(\lfloor X_1 \rceil +  \lfloor \l \rceil = i))\right|\right)\, d\l  \\
 &\leq  \int_{-\infty}^\infty f_{X_2}(\l) \frac{1}{2\s_1} \, d\l  \\
 &= \frac{1}{2\s}
  \end{align*}
  The second inequality uses Proposition \ref{prop:DDOST}.
\end{proof}

This leads to the following lemma:
\merging*
\begin{proof}
  The proof is by induction on $k$.
  The base case of $k=1$ is handled by Proposition \ref{prop:DDOSTplus}.
  For general $k$, we use a standard hybridization argument.
  Denote the $j$th coordinate of $X_i$ as $x_{ij}$.
  \begin{align*}
    &\dtv\left(\lfloor X_1 + X_2 \rceil,\lfloor X_1 \rceil + \lfloor X_2\rceil\right) \\
    &= \dtv\left((\lfloor x_{11}  + x_{21} \rceil, \dots, \lfloor x_{1k} + x_{2k}\rceil),(\lfloor x_{11}\rceil  + \lfloor x_{21} \rceil, \dots, \lfloor x_{1k}\rceil  + \lfloor x_{2k}\rceil)\right) \\
    &\leq \dtv\left((\lfloor x_{11}  + x_{21} \rceil, \dots, \lfloor x_{1k} + x_{2k}\rceil),(\lfloor x_{11}\rceil + \lfloor x_{21} \rceil, \dots,\lfloor x_{1k}  + x_{2k}\rceil)\right) \\
    &+ \dtv\left((\lfloor x_{11}\rceil + \lfloor x_{21} \rceil, \dots,\lfloor x_{1k}  + x_{2k}\rceil),(\lfloor x_{11}\rceil  + \lfloor x_{21} \rceil, \dots, \lfloor x_{1k}\rceil  + \lfloor x_{2k}\rceil)  \right) \\
    &\leq \dtv\left((\lfloor x_{11}  + x_{21} \rceil, \dots, \lfloor x_{1(k-1)} + x_{2(k-1)}\rceil),(\lfloor x_{11}\rceil + \lfloor x_{21} \rceil, \dots,\lfloor x_{1(k-1)} \rceil + \lfloor x_{2(k-1)} \rceil\right) \\
    &+ \dtv\left(\lfloor x_{1k}  + x_{2k}\rceil,\lfloor x_{1k}\rceil  + \lfloor x_{2k}\rceil\right) \\
    &\leq \frac{k-1}{2\s} + \frac{1}{2\s} = \frac{k}{2\s}
  \end{align*}
  The first inequality is the triangle inequality, the second uses Lemma \ref{lem:DPI}, and the third uses the induction hypothesis and Proposition \ref{prop:DDOSTplus}.
\end{proof}

\section{Details from Section \ref{sec:structurepmd}}
\label{sec:appendixstructure}
\subsection{Rounding the Parameters}
\label{sec:rounding}

Fix some coordinate $x$, and select all $k$-CRVs where the parameter in coordinate $x$ is in the range $(0,c)$.
Partition this subset into $k-1$ sets, depending on which coordinate $y \neq x$ is the heaviest.
We apply a rounding procedure separately to each of these sets.
After this procedure, none of the parameters in coordinate $x$ will be in $(0,c)$.
We repeat this for all $k$ possible settings of $x$.
From the description below (and the restriction that $c \leq \frac{1}{2k}$), it will be clear that we will not ``undo'' any of our work and move probabilities back into $(0,c)$, so $O(k^2)$ applications of our rounding procedure will produce the result claimed in the theorem statement.

Recall that the goal of this rounding procedure will be to shift probability mass either to or from coordinate $x$ to coordinate $y$, such that no parameter in coordinate $x$ lies in the interval $(0,c)$, while simultaneously approximately preserving the mean vector of the distribution.
We are able to do this since coordinate $y$ is ``heavy'' and thus small additions will not affect the distribution in this coordinate much.

We fix some $x,y$ in order to describe and analyze the process more formally.
Define $\mathcal{I}^x_y = \{i\, |\, 0 < \r(i,x) < c \wedge y = \arg\max_j \r(i,j)\}$ (breaking ties lexicographically), and let $M^{\r_{I^x_y}}$ be the $(n,k)$-PMD induced by this set.
For the remainder of this section, without loss of generality, assume that the indices selected by $\mathcal{I}^x_y$ are $1$ through $|\mathcal{I}^x_y|$.

Select an arbitrary set $\mathcal{R} \subseteq \mathcal{I}^x_y$ such that $|\mathcal{R}| = \left\lfloor\frac{\sum_{i' \in I^x_y} \r_{I^x_y}(i',x)}{c}\right\rfloor$.
Intuitively, this set will be the CRVs for which we set the parameter $\r(\cdot,x)$ to be $c$, while $\mathcal{I}^x_y \setminus \mathcal{R}$ will have $\r(\cdot, x)$ set to $0$.
We can perform the following rounding scheme to $\r_{I^x_y}$ to obtain a new parameter matrix $\hat \r_{I^x_y}$:
$$
\hat \r_{I^x_y}(i,j) = 
\begin{cases}
  \r_{I^x_y}(i,j) &\mbox{if } j \not\in \{x,y\} \\ \vspace{1mm}
                c &\mbox{if } j = x \wedge i \in \mathcal{R} \\ 
                0 &\mbox{if } j = x \wedge i \not \in \mathcal{R} \\
 1 - \sum_{j' \neq y} \hat \r_{I^x_y}(i,j') &\mbox{if } j = y
\end{cases}
$$


We define the process \textbf{Fork}, for sampling from a $k$-CRV $\r(i,\cdot)$ in $\mathcal{I}^x_y$:
\begin{itemize}
  \item Let $X_i$ be an indicator random variable, taking $1$ with probability $\frac{1}{k}$ and $0$ otherwise.
  \item If $X_i = 1$, then return $e_x$ with probability $k\r(i,x)$ and $e_y$ with probability $1 - k\r(i,x)$.
  \item If $X_i = 0$, then return $e_j$ with probability $0$ if $j = x$, $\frac{k}{k-1}(\r(i,x) + \r(i,y) - \frac1k)$ if $j = y$, and $\frac{k}{k-1}\r(i,j)$ otherwise.
\end{itemize}

The intuition behind this procedure is that we isolate the changes in our rounding procedure when $X_i = 1$, as when $X_i = 0$, the rounded and unrounded distributions are identical.
We note that \textbf{Fork} is well defined as long as $\r(i,x) \leq \frac{1}{k}$ and $\r(i,x) + \r(i,y) \geq \frac{1}{k}$.
The former is true since $c \leq \frac{1}{k}$, and the latter is true since $y$ was chosen to be the heaviest coordinate.
Additionally, by calculating the probability of any outcome, we can see that \textbf{Fork} is equivalent to the regular sampling process.
Define the (random) set $\bm{X} = \{i\, |\, X_i = 1\}$.
We will use $\bm{\th}$ to refer to a particular realization of this set.
We define \textbf{Fork} for sampling from $\hat \r(i,\cdot)$ in the same way, though we will denote the indicator random variables by $\hat X_i$ and $\bm{\hat X}$ instead.
Note that, if $c \leq \frac1k$, the process will still be well defined after rounding.
This is because $\hat \r(i,x) \leq c \leq \frac1k$, and $\hat \r(i,x) + \hat \r(i,y) = \r(i,x) + \r(i,y) \geq \frac1k$.
For the rest of this section, when we are drawing a sample from a CRV, we draw it via the process \textbf{Fork}.

The proof of Lemma \ref{lem:round} follows from the following three lemmata.
Intuitively, the first states that the PMD induced by the CRVs for which $X_i = 1$ gives a Poisson Binomial distribution with mean concentrated around its expected value, for both the rounded and unrounded PMDs.
The second states that if this value is concentrated, then the two distributions are close in total variation distance.
The proof relates the rounded and unrounded distributions by comparing the total variation distance between the Poisson distributions with the same means.
The third lemma eliminates the condition on the second lemma by using the first lemma, which states that this condition is likely to hold.
\begin{lemma}\label{lem:leaves}
  If $\sum_{i \in I^x_y}  \r(i,x) \geq 3 ck \log \left(\frac{1}{ck}\right)$, then
  \begin{multline*}
    \Pr\left(\bm{\th} : \vphantom{ \sum_{i \in \bm{\th}}}\right.
     \left. \left|\sum_{i \in \bm{\th}} k\r_{I^x_y}(i,j) - E\Big[\sum_{i \in \bm{X}} k\r_{I^x_y}(i,j)\Big]\right| 
    \leq \left(3ck \log{\left(\frac{1}{ck}\right)}E \Big[\sum_{i \in \bm{X}}k \r_{I^x_y}(i,j)\Big]\right)^{1/2} \right. \\ 
  \wedge 
   \left. \left|\sum_{i \in \bm{\th}} k\hat \r_{I^x_y}(i,j) - E\Big[\sum_{i \in \bm{\hat X}}k\hat  \r_{I^x_y}(i,j)\Big]\right| 
  \leq \left(3ck \log{\left(\frac{1}{ck}\right)}E \Big[\sum_{i \in \bm{\hat X}} k\hat \r_{I^x_y}(i,j)\Big]\right)^{1/2}
  \right) \\ \geq  1 - 4ck
  \end{multline*}
\end{lemma}

\begin{lemma}\label{lem:poisson}
  Suppose that, for some $\bm{\th}$, the following hold: 
$$\left|\sum_{i \in \bm{\th}} k\r_{I^x_y}(i,j) - E\Big[\sum_{i \in \bm{X}} k\r_{I^x_y}(i,j)\Big]\right| 
    \leq \left(3ck \log{\left(\frac{1}{ck}\right)}E \Big[\sum_{i \in \bm{X}}k \r_{I^x_y}(i,j)\Big]\right)^{1/2}$$

 $$ \left|\sum_{i \in \bm{\th}} k\hat \r_{I^x_y}(i,j) - E\Big[\sum_{i \in \bm{\hat X}}k\hat  \r_{I^x_y}(i,j)\Big]\right| 
  \leq \left(3ck \log{\left(\frac{1}{ck}\right)}E \Big[\sum_{i \in \bm{\hat X}} k\hat \r_{I^x_y}(i,j)\Big]\right)^{1/2}$$
Then, letting $Z_i$ be the Bernoulli random variable with expectation $k \r_{I^x_y}(i,x)$ (and $\hat Z_i$ defined similarly with $ k\hat \r_{I^x_y}(i,x)$),
  $$\mathrm{d_{TV}}\left( \sum_{i \in \bm{\th}} Z_i,\sum_{i \in \bm{\th}} \hat Z_i  \right)
  < O\left(c^{1/2} k^{1/2} \log^{1/2}\left(\frac{1}{ck}\right)\right)$$
\end{lemma}

\begin{lemma}\label{lem:glue}
  For any $\mathcal{I}^x_y$,
  $$\dtv\left(M^{\r_{I^x_y}},M^{\hat\r_{I^x_y}}\right) < O\left(c^{1/2} k^{1/2} \log^{1/2}\left(\frac{1}{ck}\right)\right)$$
\end{lemma}
Since our final rounded $(n,k)$-PMD is generated after applying this rounding procedure $O(k^2)$ times, Lemma \ref{lem:round} follows from our construction and Lemma \ref{lem:glue} via the triangle inequality.

\begin{prevproof}{Lemma}{lem:leaves}
  Note that $\sum_{i \in \bm{X}} k\r_{I^x_y}(i,x) = \sum_{i \in I^x_y} \Om_i$, where 
  $$\Om_i = 
  \begin{cases}
    k\r_{I^x_y}(i,x) & \mbox{with probability } \frac1k \\
                   0 & \mbox{with probability } 1 - \frac1k
  \end{cases}
  $$
  We apply Lemma \ref{lem:cher} to the rescaled random variables $\Om_i' = \frac{1}{ck}\Om_i$, with $\g = \sqrt{\frac{3\log{\frac1{ck}}}{E[\sum_{i \in I^x_y} \Om_i']}}$, giving
  $$\Pr\left[\left|\sum_{i \in I^x_y} \Om_i' - E\Big[\sum_{i \in I^x_y} \Om_i'\Big]\right| \geq \left(3 \log \left(\frac1{ck}\right)E\Big[\sum_{i \in I^x_y} \Om_i'\Big] \right)^{1/2}\right] \leq 2ck.$$
  Unscaling the variables gives
  \begin{multline*}
  \Pr\left[\left|\sum_{i \in \bm{X}} k\r_{I^x_y}(i,x) - E\Big[\sum_{i \in \bm{X}} k\r_{I^x_y}(i,x) \Big]\right| \geq \left(3ck \log \left(\frac1{ck}\right) E\Big[\sum_{i \in \bm{X}} k\r_{I^x_y}(i,x) \Big]\right)^{1/2}\right] \\ \leq 2ck.
  \end{multline*}

  Applying the same argument to $\hat \r_{I^x_y}$ gives 
  \begin{multline*}
  \Pr\left[\left|\sum_{i \in \bm{\hat X}} k\hat \r_{I^x_y}(i,x) - E\Big[\sum_{i \in \bm{\hat X}} k\hat \r_{I^x_y}(i,x) \Big]\right| \geq \left(3ck \log \left(\frac1{ck}\right) E\Big[\sum_{i \in \bm{\hat X}} k\hat \r_{I^x_y}(i,x) \Big]\right)^{1/2}\right]\\ \leq 2ck.
  \end{multline*}

  Since $\bm{X} \sim \bm{\hat X}$, by considering the joint probability space where $\bm{\th} = \bm{X} = \bm{\hat X}$ and applying a union bound, we get
  \begin{multline*}
    \Pr\left(\bm{\th} : \vphantom{ \sum_{i \in \bm{\th}}}\right.
     \left. \left|\sum_{i \in \bm{\th}} k\r_{I^x_y}(i,j) - E\Big[\sum_{i \in \bm{X}} k\r_{I^x_y}(i,j)\Big]\right| 
    \leq \left(3ck \log{\left(\frac{1}{ck}\right)}E \Big[\sum_{i \in \bm{X}}k \r_{I^x_y}(i,j)\Big]\right)^{1/2} \right. \\ 
  \wedge 
   \left. \left|\sum_{i \in \bm{\th}} k\hat \r_{I^x_y}(i,j) - E\Big[\sum_{i \in \bm{\hat X}}k\hat  \r_{I^x_y}(i,j)\Big]\right| 
  \leq \left(3ck \log{\left(\frac{1}{ck}\right)}E \Big[\sum_{i \in \bm{\hat X}} k\hat \r_{I^x_y}(i,j)\Big]\right)^{1/2}
  \right)\\ \geq 1 - 4ck.
  \end{multline*}
\end{prevproof}

\begin{prevproof}{Lemma}{lem:poisson}
  Fix some $\bm{\th} = \bm{X} = \bm{\hat X}$.
  Without loss of generality, assume $E \Big[\sum_{i \in \bm{X}}k \r_{I^x_y}(i,j)\Big] \geq E \Big[\sum_{i \in \bm{\hat X}}k\hat  \r_{I^x_y}(i,j)\Big]$.
  There are two cases:
  \begin{case}
    $ E \Big[\sum_{i \in \bm{X}}k \r_{I^x_y}(i,j)\Big] \leq (ck)^{3/4}$
  \end{case}
  From the first assumption in the lemma statement,
  \begin{align*}
    \sum_{i \in \bm{\th}} k\r_{I^x_y}(i,j) &\leq  E\Big[\sum_{i \in \bm{X}} k\r_{I^x_y}(i,j)\Big]
    + \left(3ck \log{\left(\frac{1}{ck}\right)}E \Big[\sum_{i \in \bm{X}}k \r_{I^x_y}(i,j)\Big]\right)^{1/2} \\
    &\leq (ck)^{3/4} + \sqrt{3}(ck)^{7/8}\log^{1/2}{\left(\frac{1}{ck}\right)} \coloneqq g(c,k)
  \end{align*}
  Similarly, by the second assumption in the lemma statement and since \\
  $E \Big[\sum_{i \in \bm{X}}k \r_{I^x_y}(i,j)\Big] \geq E \Big[\sum_{i \in \bm{\hat X}}k\hat  \r_{I^x_y}(i,j)\Big],$
  we also have that
  $\sum_{i \in \bm{\th}} k\hat \r_{I^x_y}(i,j) \leq g(c,k)$.

  By Markov's inequality, $\Pr\Big[\sum_{i \in \bm{\th}} Z_i \geq 1\Big] \leq  \sum_{i \in \bm{\th}} k\r_{I^x_y}(i,j) \leq g(c,k)$,
  and similarly, $\Pr\Big[\sum_{i \in \bm{\th}} \hat Z_i \geq 1\Big] \leq g(c,k)$.
  This implies that $$\left| \Pr\Big[\sum_{i \in \bm{\th}} Z_i = 0\Big] - \Pr\Big[\sum_{i \in \bm{\th}} \hat Z_i = 0\Big] \right| \leq 2g(c,k),$$
  and thus by the coupling lemma,
  $$\dtv\left(\sum_{i \in \bm{\th}} Z_i,\sum_{i \in \bm{\th}} \hat Z_i\right) \leq 4g(c,k) = 4\left((ck)^{3/4} + \sqrt{3}(ck)^{7/8}\log^{1/2}{\left(\frac{1}{ck}\right)} \right)$$

  \begin{case}
    $ E \Big[\sum_{i \in \bm{X}}k \r_{I^x_y}(i,j)\Big] \geq (ck)^{3/4}$
  \end{case}
  We use the following proposition, which is a combination of a classical result in Poisson approximation \cite{BarbourHJ92} and Lemma 3.10 in \cite{DaskalakisP07}.
  \begin{proposition}
    For any set of independent Bernoulli random variables $\{Z_i\}_i$ with expectations $E[Z_i] \leq ck$,
    $$\dtv\left(\sum_i Z_i,Poisson\Big(E\Big[\sum_i Z_i\Big]\Big)\right) \leq ck.$$
  \end{proposition}
  Applying this, we see
  $$\dtv\left(\sum_{i \in \bm{\th}} Z_i,Poisson\Big(E\Big[\sum_{i \in \bm{\th}} Z_i\Big]\Big)\right) \leq ck$$
  $$\dtv\left(\sum_{i \in \bm{\th}} \hat Z_i,Poisson\Big(E\Big[\sum_{i \in \bm{\th}} \hat Z_i\Big]\Big)\right) \leq ck$$
  We must now bound the distance between the two Poisson distributions.
  We use the following lemma from \cite{DaskalakisP08}:
  \begin{lemma}[Lemma B.2 in \cite{DaskalakisP08}]
    If $\l = \l_0 + D$ for some $D > 0, \l_0 > 0$,
    $$\dtv\left(Poisson(\l),Poisson(\l_0)\right) \leq D\sqrt{\frac{2}{\l_0}}.$$
  \end{lemma} 
  Applying this gives that 
  \begin{align*}
    &\dtv\left(Poisson\Big(E\Big[\sum_{i \in \bm{\th}} Z_i\Big]\Big),Poisson\Big(E\Big[\sum_{i \in \bm{\th}} \hat Z_i\Big]\Big)\right) \\
    &\leq
  \left|E\Big[\sum_{i \in \bm{\th}} Z_i\Big] - E\Big[\sum_{i \in \bm{\th}}\hat  Z_i\Big]\right| 
  \sqrt{\frac{2}{\min\left\{E\Big[\sum_{i \in \bm{\th}} Z_i\Big],E\Big[\sum_{i \in \bm{\th}}\hat Z_i\Big]\right\}}}
  \end{align*}

  To bound this, we need the following proposition, which we prove below:
  \begin{proposition}\label{prop:meanbound}
  $$\sqrt{\frac{2\left|\sum_{i \in \bm{\th}} k\r_{I^x_y}(i,x) - \sum_{i \in \bm{\th}} k\hat \r_{I^x_y}(i,x)\right|^2}
  {\min\left\{
    \sum_{i \in \bm{\th}} k \r_{I^x_y}(i,x),\sum_{i \in \bm{\th}} k\hat \r_{I^x_y}(i,x)
  \right\}}} \leq \sqrt{80ck\log\left(\frac{1}{ck}\right)}$$
  \end{proposition}

  Thus, using the triangle inequality and this proposition, for sufficiently small $c$, we get
  $$\dtv\left(\sum_{i \in \bm{\th}} Z_i,\sum_{i \in \bm{\th}} \hat Z_i\right) \leq 2ck + \sqrt{80ck\log\left(\frac{1}{ck}\right)} = O\left(c^{1/2}k^{1/2}\log^{1/2}\left(\frac1{ck}\right)\right).$$

  By comparing Cases 1 and 2, we see that the desired bound holds in both cases.

  \begin{prevproof}{Proposition}{prop:meanbound}
    By the definition of our rounding procedure, we observe that 
    $$\left|E\Big[\sum_{i \in \bm{X}} k\r_{I^x_y}(i,x) \Big] - E\Big[\sum_{i \in \bm{\hat X}} k\hat \r_{I^x_y}(i,x) \Big]\right| \leq c$$
    By the assumptions of Lemma \ref{lem:poisson} and the assumption that $E \Big[\sum_{i \in \bm{X}}k \r_{I^x_y}(i,j)\Big] \geq E \Big[\sum_{i \in \bm{\hat X}}k\hat  \r_{I^x_y}(i,j)\Big]$,
    \begin{align*}
      \left|\sum_{i \in \bm{\th}} k\r_{I^x_y}(i,x) - \sum_{i \in \bm{\th}} k\hat \r_{I^x_y}(i,x)\right|  
      &\leq \left|E\Big[\sum_{i \in \bm{X}} k\r_{I^x_y}(i,x) \Big] - E\Big[\sum_{i \in \bm{\hat X}} k\hat \r_{I^x_y}(i,x) \Big]\right| \\
      &+ \left(12ck\log\left(\frac{1}{ck}\right) E\Big[\sum_{i \in \bm{X}} k\r_{I^x_y}(i,x) \Big]\right)^{1/2} \\
      &\leq c + \left(12ck\log\left(\frac{1}{ck}\right) E\Big[\sum_{i \in \bm{X}} k\r_{I^x_y}(i,x) \Big]\right)^{1/2},
    \end{align*}
    and thus,
    \begin{align} 
      &\left|\sum_{i \in \bm{\th}} k\r_{I^x_y}(i,x) - \sum_{i \in \bm{\th}} k\hat \r_{I^x_y}(i,x)\right|^2 \nonumber \\
      &\leq c^2 + 12ck\log\left(\frac{1}{ck}\right) E\Big[\sum_{i \in \bm{X}} k\r_{I^x_y}(i,x) \Big] + \left(48c^3k\log\left(\frac{1}{ck}\right) E\Big[\sum_{i \in \bm{X}} k\r_{I^x_y}(i,x) \Big]\right)^{1/2} \label{eq:num}
    \end{align}
    From the assumption that $E\Big[\sum_{i \in \bm{X}} k\r_{I^x_y}(i,x) \Big] \geq (ck)^{3/4}$, for sufficiently small $c$,
    \begin{align*}
      E\Big[\sum_{i \in \bm{X}} k\r_{I^x_y}(i,x) \Big] &\geq (ck)^{3/8}\left(E\Big[\sum_{i \in \bm{X}} k\r_{I^x_y}(i,x) \Big]\right)^{1/2} \\
                                                       &\geq \left(12ck \log{\left(\frac{1}{ck}\right)}E \Big[\sum_{i \in \bm{X}}k \r_{I^x_y}(i,j)\Big]\right)^{1/2}
    \end{align*}

    Combining this with the first assumption of Lemma \ref{lem:poisson}, 
    \begin{align*}
      \sum_{i \in \bm{\th}} k\r_{I^x_y}(i,x) &\geq E\Big[\sum_{i \in \bm{X}} k\r_{I^x_y}(i,x) \Big] - \left(3ck \log{\left(\frac{1}{ck}\right)}E \Big[\sum_{i \in \bm{X}}k \r_{I^x_y}(i,j)\Big]\right)^{1/2} \\
  &\geq \frac12 E\Big[\sum_{i \in \bm{X}} k\r_{I^x_y}(i,x) \Big] 
  \end{align*}
  Similarly, since 
  $E \Big[\sum_{i \in \bm{\hat X}}k\hat  \r_{I^x_y}(i,j)\Big] \geq E \Big[\sum_{i \in \bm{X}}k \r_{I^x_y}(i,j)\Big] - c \geq (ck)^{3/4} - c$, for $c$ sufficiently small,
  $$\sum_{i \in \bm{\th}} k\hat \r_{I^x_y}(i,x) \geq \frac12 E\Big[\sum_{i \in \bm{X}} k\hat \r_{I^x_y}(i,x) \Big]$$

  It follows that
  \begin{align}
    &\min\left\{\sum_{i \in \bm{\th}} k \r_{I^x_y}(i,x),\sum_{i \in \bm{\th}} k\hat \r_{I^x_y}(i,x) \right\}\nonumber  \\
    &\geq \frac12 \min\left\{E\Big[\sum_{i \in \bm{X}} k \r_{I^x_y}(i,x)\Big],E\Big[\sum_{i \in \bm{\hat X}} k\hat \r_{I^x_y}(i,x)\Big] \right\}\nonumber  \\
    &= \frac12 E\Big[\sum_{i \in \bm{\hat X}} k\hat \r_{I^x_y}(i,x)\Big] \nonumber \\
    &\geq \frac12 \left(E\Big[\sum_{i \in \bm{X}} k\r_{I^x_y}(i,x)\Big] - c\right) \nonumber \\
    &\geq \frac14 E\Big[\sum_{i \in \bm{X}} k\r_{I^x_y}(i,x)\Big] \label{eq:den}
  \end{align}
  where the last equality follows for $c$ sufficiently small because $E\Big[\sum_{i \in \bm{X}} k\r_{I^x_y}(i,x)\Big] \geq (ck)^{3/4}$.

  From (\ref{eq:num}) and (\ref{eq:den}), for $c$ sufficiently small,
  $$\frac{2\left|\sum_{i \in \bm{\th}} k\r_{I^x_y}(i,x) - \sum_{i \in \bm{\th}} k\hat \r_{I^x_y}(i,x)\right|^2}
  {\min\left\{
    \sum_{i \in \bm{\th}} k \r_{I^x_y}(i,x),\sum_{i \in \bm{\th}} k\hat \r_{I^x_y}(i,x)
  \right\}} \leq 80ck\log\left(\frac{1}{ck}\right),$$
  from which the proposition statement follows.
  \end{prevproof}
\end{prevproof}

\begin{prevproof}{Lemma}{lem:glue}
  First, note that if $\sum_{i \in I^x_y}  \r(i,x) < 3 ck \log \left(\frac{1}{ck}\right)$, The $\ell_1$ distance between the parameters of the rounded and the unrounded distributions is at most $6 ck \log \left(\frac{1}{ck}\right)$.
  By the triangle inequality and the Data Processing Inequality (Lemma \ref{lem:DPI}), this is an upper bound for the total variation distance between the rounded and unrounded distributions, and the desired conclusion holds.
  Therefore, for the remainder of the proof, assume that $\sum_{i \in I^x_y}  \r(i,x) \geq 3ck \log \left(\frac{1}{ck}\right)$.

  Throughout this proof, we will couple the two sampling processes such that $\bm{\th} \coloneqq\bm{X} = \bm{\hat X} $, which is possible since $\bm{X} \sim \bm{\hat X}$.
  Let $\phi$ be the random event that $\bm{\th}$ satisfies the following conditions:
  \begin{align*}
    & \left. \left|\sum_{i \in \bm{\th}} k\r_{I^x_y}(i,j) - E\Big[\sum_{i \in \bm{X}} k\r_{I^x_y}(i,j)\Big]\right| 
    \leq \left(3ck \log{\left(\frac{1}{ck}\right)}E \Big[\sum_{i \in \bm{X}}k \r_{I^x_y}(i,j)\Big]\right)^{1/2} \right. \\ 
  & \left|\sum_{i \in \bm{\th}} k\hat \r_{I^x_y}(i,j) - E\Big[\sum_{i \in \bm{\hat X}}k\hat  \r_{I^x_y}(i,j)\Big]\right| 
  \leq \left(3ck\log{\left(\frac{1}{ck}\right)}E \Big[\sum_{i \in \bm{\hat X}} k\hat \r_{I^x_y}(i,j)\Big]\right)^{1/2}
  \end{align*}
 
  Suppose that $\phi$ occurs, and fix a $\bm{\th}$ in this probability space.
  We start by showing that for such a $\bm{\th}$, 
  $$\mathrm{d_{TV}}\left(M^{\r_{I^x_y}},M^{\hat\r_{I^x_y}}\, \middle|\, \bm{X} = \bm{\hat X} = \bm{\th}\right) < O\left(c^{1/2} k^{1/2} \log^{1/2}\left(\frac{1}{ck}\right)\right)$$
  Let $M^{\r_{I^x_y}^{\bm{\th}}}$ and $M^{\r_{I^x_y}^{\bm{\bar\th}}}$ be the $(n,k)$-PMDs induced by the $k$-CRVs in $M^{\r_{I^x_y}}$ with indices in $\bm{\th}$ and not in $\bm{\th}$, respectively.
  Define $M^{\hat \r_{I^x_y}^{\bm{\th}}}$ and $M^{\hat \r_{I^x_y}^{\bm{\bar\th}}}$ similarly.
  We can see
  \begin{align*}
  \mathrm{d_{TV}}\left(M^{\r_{I^x_y}},M^{\hat\r_{I^x_y}}\, \middle|\, \bm{X} = \bm{\hat X} = \bm{\th}\right)
  &= \mathrm{d_{TV}}\left(M^{\r_{I^x_y}^{\bm{\th}}} + M^{\r_{I^x_y}^{\bm{\bar\th}}},
                         M^{\hat \r_{I^x_y}^{\bm{\th}}} + M^{\hat \r_{I^x_y}^{\bm{\bar\th}}}
  \middle|\, \bm{X} = \bm{\hat X} = \bm{\th}\right) \\
  &\leq \mathrm{d_{TV}}\left(M^{\r_{I^x_y}^{\bm{\th}}}, M^{\hat \r_{I^x_y}^{\bm{\th}}}  \middle|\, \bm{X} = \bm{\hat X} = \bm{\th}\right) \\
  &+ \mathrm{d_{TV}}\left(M^{\r_{I^x_y}^{\bm{\bar \th}}}, M^{\hat \r_{I^x_y}^{\bm{\bar \th}}}  \middle|\, \bm{X} = \bm{\hat X} = \bm{\th}\right) \\
  &\leq \mathrm{d_{TV}}\left(M^{\r_{I^x_y}^{\bm{\th}}}, M^{\hat \r_{I^x_y}^{\bm{\th}}}  \middle|\, \bm{X} = \bm{\hat X} = \bm{\th}\right) \\
  &= \mathrm{d_{TV}}\left(\sum_{i \in \bm{\th}} Z_i,\sum_{i \in \bm{\th}} \hat Z_i  \middle|\, \bm{X} = \bm{\hat X} = \bm{\th}\right) \\
  &\leq O\left(c^{1/2}k^{1/2}\log^{1/2}\left(\frac1{ck}\right)\right)
  \end{align*}
  The first inequality is the triangle inequality,
  the second inequality is because the distributions for $k$-CRVs in $\bm{\bar \th}$ are identical (since we do not change them in our rounding),
  and the third inequality is Lemma \ref{lem:poisson}.

  By the law of total probability for total variation distance,
  \begin{align*}
  \mathrm{d_{TV}}\left(M^{\r_{I^x_y}},M^{\hat\r_{I^x_y}}\right) 
  &= \Pr(\phi)\mathrm{d_{TV}}\left(M^{\r_{I^x_y}},M^{\hat\r_{I^x_y}} \middle| \phi \right)
  + \Pr(\bar \phi)\mathrm{d_{TV}}\left(M^{\r_{I^x_y}},M^{\hat\r_{I^x_y}} \middle| \bar \phi \right) \\
  &\leq \left(1 - 4ck\right) \cdot O\left(c^{1/2}k^{1/2}\log^{1/2}\left(\frac1{ck}\right)\right) + 4ck \cdot 1 \\
  &= O\left(c^{1/2}k^{1/2}\log^{1/2}\left(\frac1{ck}\right)\right)
  \end{align*}
  where the inequality is obtained by applying Lemma \ref{lem:leaves} and the bound shown above pointwise for $\bm{\th}$ which satisfy $\phi$.

\end{prevproof}

\subsection{Converting to a Discretized Gaussian using the Valiant-Valiant CLT}
\label{sec:vvclt}
We will now apply a result by Valiant and Valiant \cite{ValiantV10}.
We recall the aforementioned CLT by Valiant and Valiant, Theorem \ref{thm:val}, which we restate for convenience.

\val*

As we can see from this inequality, there are two issues that may arise and lead to a bad approximation:
\begin{itemize}
  \item $G^\r$ has small variance in some direction (cf. Proposition \ref{prop:eigcov})
  \item $G^\r$ has a large size parameter $n$
\end{itemize}
We must avoid both of these issues simultaneously -- we will apply this result to several carefully chosen sets, and then merge the resulting Gaussians into one using Lemma \ref{lem:merge}.

The first step is to partition our CRVs into several sets, and then convert the PMDs induced by each set into GMDs (with an appropriately chosen pivot).
The original PMD can be sampled by sampling each of these GMDs and then adding their results.
In other words, the probability mass function of the PMD is the convolution of the probability mass functions of these GMDs.

We start by partitioning the $k$-CRVs into $k$ sets $S_1, \dots, S_k$, where $S_{j'}= \{i\, |\, j' = \arg\max_j \hat \p(i,j) \}$ and ties are broken by lexicographic ordering.
This defines $S_{j'}$ to be the set of indices of $k$-CRVs in which $j'$ is the heaviest coordinate.
Let $M^{\hat \p_{j'}}$ be the $(|S_{j'}|,k)$-PMD induced by taking the $k$-CRVs in $S_{j'}$.
For the remainder of this section, we will focus on $S_k$, the other cases follow symmetrically.

We convert each CRV in $S_k$ into a truncated $k$-CRV by omitting the $k$th coordinate, giving us a $(|S_k|,k)$-GMD $G^{\hat \r_{k}}$.
Since the $k$th coordinate was the heaviest, we can make the following observation:
\begin{observation}\label{obs:heavycoord}
  $\hat \r_k(i,0) \geq \frac1k$ for all $i \in S_k$.
\end{observation}

If we tried to apply Theorem \ref{thm:val} to $G^{\hat \r_k}$, we would obtain a vacuous result.
For instance, if there exists a $j$ such that $\hat \r_k(i,j) = 0$ for all $i$, the variance in this direction would be $0$ and Theorem \ref{thm:val} would give us a trivial result.
Therefore, we further partition $S_k$ into $2^{k-1}$ sets indexed by $2^{[k-1]}$, where each set contains the elements of $S_k$ which are non-zero on its indexing set and zero otherwise.
More formally, $S_k^{\mathcal{I}} = \{i\,|\, (i \in S_k) \wedge (\hat \r(i,j) \geq c\ \forall j \in \mathcal{I}) \wedge (\hat \r(i,j) = 0\ \forall j \not\in \mathcal{I}) \}$.
For each of these sets, due to our rounding procedure, we know that the variance is non-negligible in each of the non-zero directions.
Naively, we would apply the CLT separately to each of these sets.
The issue is that merging the resulting $2^k$ Gaussians would be costly.
Roughly, merging two Gaussians into one incurs a cost proportional to the inverse of the minimum standard deviation of either Gaussian.
In order to avoid the cost of merging exponentially many Gaussians with similar variances, before applying the CLT, we group sets $S_k^{\mathcal{I}}$ of similar variance together.
The resulting collection of Gaussians have variances which increase rapidly, and by merging them in the correct order, we can minimize this cost.

Recall that $\g = O(1)$ and $t = \poly(k/\ve)$ (as specified in Section \ref{sec:params}).
For an integer $l \geq 0$, define $B^l = \bigcup_{\mathcal{I} \in Q_l} S_k^{\mathcal{I}}$, where $Q_l = \{\mathcal{I}\, |\, |S_k^{\mathcal{I}}| \in [l^\g t,(l+1)^\g t)\}$.
In other words, bucket $l$ will contain a collection of truncated CRVs, defined by the union of the previously defined sets which have a size falling in a particular interval.

At this point, we are ready to apply the central limit theorem:
\begin{lemma}\label{lem:densebin}
  Let $G^{\hat \r_k^l}$ be the $(|B^l|,k)$-GMD induced by the truncated CRVs in $B^l$, and $\m_k^l$ and $\S_k^l$ be its mean and covariance matrix.
  Then 
  $$\dtv\left(G^{\hat \r_k^l},\lfloor \mathcal{N}(\m_k^l,\S_k^l)\rceil\right) \leq \frac{8.646k^{3/2}\log^{2/3}(2^k(l+1)^\g t)}{l^{\g/6}t^{1/6}c^{1/6}}.$$
  Furthermore, the minimum non-zero eigenvalue of $\S_k^l$ is at least $\frac{l^\gamma t c}{k}$.
\end{lemma}
\begin{proof}
  This follows from Theorem \ref{thm:val}, it suffices to bound the values of ``$n$'' and ``$\s^2$'' which appear in the theorem statement.

  $B^l$ is the union of at most $2^k$ sets, each of size at most $(l+1)^\g t$, which gives us the upper bound of $2^k (l+1)^\g t$ as the size of induced GMD.

  We must be more careful when reasoning about the minimum eigenvalue of $\S_k^l$ -- indeed, it may be $0$ if there exists a $j'$ such that for all $i$, $\hat \r_k^l(i,j') = 0$.
  Therefore, we apply the CLT on the GMD defined by removing all zero-columns from $\r_k^l$, taking us down to a dimension $k' \leq k$.
  Afterwards, we lift the related discretized Gaussian up to $k$ dimensions by inserting $0$ for the means and covariances involving any of the $k - k'$ dimensions we removed.
  This operation will not increase the total variation distance, by Lemma \ref{lem:DPI}.
  From this point, we assume that all columns of $\hat \r_k^l$ are non-zero.

  Consider an arbitrary $S_k^\mathcal{I}$ which is included in $B^l$.
  Let $\mathcal{E_I} = \mathrm{span}\{e_i\, |\, i \in \mathcal{I}\}$.
  Applying Proposition \ref{prop:eigcov}, Observation \ref{obs:heavycoord}, and the properties necessary for inclusion in $S_k^\mathcal{I}$, we can see that a CRV in $S_k^\mathcal{I}$ has variance at least $\frac{c}{k}$ within $\mathcal{E_I}$.
  Since inclusion in $B^l$ means that $|S_k^\mathcal{I}| \geq l^\g t$, and variance is additive for independent random variables, the GMD induced by $S_k^\mathcal{I}$ has variance at least $l^\g t\frac{c}{k}$ within $\mathcal{E_I}$.
  To conclude, we note that if a column in $\hat \r_k^l$ is non-zero, there must be some $\mathcal{I^*} \in Q_l$ which intersects the corresponding dimension.
  Since $S_k^\mathcal{I^*}$ causes the variance in this direction to be at least $l^\g t\frac{c}{k}$, we see that the variance in every direction must be this large.
  This also implies the bound on the minimum non-zero eigenvalue of $\S_k^l$. 

  By substituting these values into Theorem \ref{thm:val}, we obtain the claimed bound.
\end{proof}

We note that this gives us a vacuous bound for $B^0$, which we must deal with separately.
The issue with this bucket is that the variance in some directions might be small compared to the size of the GMD induced by the bucket.
The intuition is that we can remove the truncated CRVs which are non-zero in these low-variance dimensions, and the remaining truncated CRVs can be combined into another GMD.

\sparsebin*
\begin{proof}
  The algorithm iteratively eliminates columns which have fewer than $t$ non-zero entries.
  For each such column $j$, add all truncated CRVs which have non-zero entries in column $j$ to $\bar S$.
  Since there are only $k$ columns, we add at most $kt$ truncated CRVs to $\bar S$.

  Now, we apply Theorem \ref{thm:val} to the truncated CRVs in $S$.
  The analysis of this is similar to the proof of Lemma \ref{lem:densebin}.
  As argued before, we can drop the dimensions which have $0$ variance.
  This time, the size of the GMD is at most $2^k t$, which follows from the definition of $B^0$.
  Recall that the minimum variance of a single truncated CRV in $S$ is at least $\frac{c}{k}$ in any direction in the span of its non-zero columns.
  After removing the CRVs in $\bar S$, every dimension with non-zero variance must have at least $t$ truncated CRVs which are non-zero in that dimension, giving a variance of at least $\frac{tc}{k}$.
  Substituting these parameters into Theorem \ref{thm:val} gives the claimed bound.
\end{proof}

We assemble the two lemmata to obtain the following result:
\begin{lemma}\label{lem:singlestruct}
  Let $G^{\hat \r_k}$ be a $(n,k)$-GMD with $\hat \r_k(i,j) \not\in (0,c)$ and $\sum_j \r_k(i,j) \leq 1 - \frac1k$ for all $i$, and let $S_k$ be its set of component truncated CRVs.
  There exists an efficiently computable partition of $S_k$ into $S$ and $\bar S$, where $|\bar S| \leq kt$.
  Furthermore, letting $\m_S$ and $\S_S$ be the mean and covariance matrix of the $(|S|,k)$-GMD induced by $S$, and $G^{\hat \r_k^{\bar S}}$ be the $(|\bar S|, k)$-GMD induced by $\bar S$,
  $$\dtv\left(G^{\hat \r_k},\lfloor \mathcal{N}(\m_S,\S_S)\rceil \ast G^{\hat \r_k^{\bar S}}\right) \leq O\left(\frac{k^{13/6}\log^{2/3}t}{c^{1/6}t^{1/6}} + \frac{k^{3/2}}{c^{1/2}t^{1/2}}\right).$$
  Furthermore, the minimum non-zero eigenvalue of $\S_S$ is at least $\frac{t c}{k}$.
\end{lemma}
\begin{proof}
  This is a combination of Lemmas \ref{lem:densebin} and \ref{lem:sparsebin}, with the results merged using Lemma \ref{lem:merge}.

  As described above, we will group the truncated CRVs into several buckets.
  We first apply Lemma \ref{lem:densebin} to each of the non-empty buckets $B^l$ for $l > 0$.
  This will give us a sum of many discretized Gaussians.
  If applicable, we apply Lemma \ref{lem:sparsebin} to $B^0$ to obtain another discretized Gaussian and a set $\bar S$ of $\leq kt$ truncated CRVs.
  By applying Lemma \ref{lem:merge}, we can ``merge'' the sum of many discretized Gaussians into a single discretized Gaussian.
  By triangle inequality, the error occured in the theorem statement is the sum of all of these approximations.
  
  We start by analyzing the cost of applying Lemma \ref{lem:densebin}.
  Recall $\g = 6 + \d_\g$ for some constant $\d_\g > 0$.
  Let the set of $N$ non-empty buckets be $\mathcal{X}$.
  Then the sum of the errors incurred by all $N$ applications of Lemma \ref{lem:densebin} is at most
  \begin{align*}
    \sum_{l \in \mathcal{X}} O\left(\frac{k^{3/2}\log^{2/3}(2^k(l+1)^{(6 + \d_\g)} t)}{l^{(6 + \d_\g)/6}t^{1/6}c^{1/6}}\right)
    &\leq \sum_{l=1}^\infty O\left(\frac{k^{3/2}\log^{2/3}(2^k(l+1)^{(6 + \d_\g)} t)}{l^{(6 + \d_\g)/6}t^{1/6}c^{1/6}}\right) \\
    &\leq \sum_{l=1}^\infty O\left(\frac{k^{13/6}\log^{2/3}l\log^{2/3}t}{l^{(6 + \d_\g)/6}t^{1/6}c^{1/6}}\right) \\
    &\leq \frac{k^{13/6}\log^{2/3}t}{c^{1/6}t^{1/6}} \sum_{l=1}^\infty O\left(\frac{\log^{2/3}l}{l^{(6 + \d_\g)/6}}\right) \\
    &\leq \frac{k^{13/6}\log^{2/3}t}{c^{1/6}t^{1/6}} \sum_{l=1}^\infty O\left(\frac{1}{l^{(6 + \d')/6}}\right) \\
    &\leq O\left(\frac{k^{13/6}\log^{2/3}t}{c^{1/6}t^{1/6}}\right)
  \end{align*}
  for any constant $0 < \d' < \d_\g$.
  The final inequality is because the series $\sum_{n=1}^\infty n^{-c}$ converges for any $c > 1$.

  The cost of applying Lemma \ref{lem:sparsebin} is analyzed similarly,
  $$\frac{8.646k^{3/2}\log^{2/3}(2^k t)}{t^{1/6}c^{1/6}} \leq O\left(\frac{k^{13/6}\log^{2/3}t}{c^{1/6}t^{1/6}}\right)$$

  Finally, we analyze the cost of merging the $N+1$ Gaussians into one.
  We will analyze this by considering the following process:
  we maintain a discretized Gaussian, which we will name the candidate.
  The candidate is initialized to be the Gaussian generated from the highest numbered non-empty bucket.
  At every time step, we update the candidate to be the result of merging itself with the Gaussian from the highest numbered non-empty bucket which has not yet been merged.
  We continue until the Gaussian from every non-empty bucket has been merged with the candidate.

  By Lemma \ref{lem:merge}, the cost of merging two Gaussians is at most $O\left(\frac{k}{\s}\right)$, where $\s^2$ is the minimum variance of either Gaussian in any direction where either has a non-zero variance.
  From Lemma \ref{lem:densebin}, the variance of the Gaussian from $B^l$ is at least $l^\g t\frac{c}{k}$ in every direction of non-zero variance.
  Since we are considering the buckets in decreasing order and merging two Gaussians only increases the variance, when merging the candidate with bucket $l$, the maximum cost we can incur is $\left(\frac{k^{3/2}}{l^{\g/2}c^{1/2}t^{1/2}}\right)$.
  Summing over all buckets in $\mathcal{X}$,
  \begin{align*}
    \sum_{l \in \mathcal{X}} O\left(\frac{k^{3/2}}{l^{\g/2}c^{1/2}t^{1/2}}\right)
    &\leq \frac{k^{3/2}}{c^{1/2}t^{1/2}} \sum_{l=1}^\infty O\left(\frac{1}{l^{(6 + \d_\g)/2}}\right) \\
    &\leq O\left(\frac{k^{3/2}}{c^{1/2}t^{1/2}}\right)
  \end{align*}
  where the second inequality is because the series $\sum_{n=1}^\infty n^{-c}$ converges for any $c > 1$.
  We note that, from Lemma \ref{lem:sparsebin}, the variance of the Gaussian obtained from $B^0$ is at least $\frac{tc}{k}$ in any non-zero direction.
  Therefore, merging this Gaussian with the rest does not affect our bound asymptotically.
  Since the minimum non-zero variance of any Gaussian we merged was at least $\frac{tc}{k}$, the same holds for the resulting merged Gaussian and its minimum non-zero eigenvalue.

  By adding the error terms obtained from each of the approximations, we obtain the claimed bound on total variation distance.
\end{proof}

\subsection{Merging $k$ Gaussians into one}
\label{sec:merging}

In order to merge the $k$ discretized Gaussians into one, we perform a series of ``swap-and-merge'' operations, in which we swap the pivots of two discretized Gaussians to be the same, and then merge the resulting distributions into one.
We repeat this process until all Gaussians which overlap in some dimension are merged together.
The following lemma bounds the cost of swapping a pivot.

\swaptv*

\begin{proof}
Without loss of generality, assume $(i,j) = (1,2)$ and $\sigma_{-2}^2 \leq \sigma_{-1}^2$. 
Sampling from $Y_1$ can be described by the following process:
Draw a sample $x_{-1,2}^{(1)} \sim \mathcal{N}(\mu_{-1,2},\Sigma_{-1,2})$, which is the Gaussian obtained from $\mathcal{N}(\mu,\Sigma)$ by projecting on to all dimensions except $1$ and $2$.
Now, condition on $x_{-1,2}^{(1)}$, sample the 2nd coordinate $x_2^{(1)}$ as the one dimensional projection onto $e_2$ of $\mathcal{N}(\mu,\Sigma)$ conditioned on $x_{-1,2}^{(1)}$, discretize all these values (i.e., round them to the nearest integer), and then set the 1st coordinate to be $\lfloor x_1^{(1)}\rceil = n - \sum_{\ell \geq 3} \lfloor x_\ell^{(1)} \rceil - \lfloor x_2^{(1)} \rceil$. 

Similarly, to draw a sample from $Y_2$, we first sample $x_{-1,2}^{(2)} \sim \mathcal{N}(\m_{-1,2},\S_{-1,2})$. 
We then condition on $x_{-1,2}^{(2)}$, sample the 1st coordinate $x_1^{(1)}$ as the one dimensional projection onto $e_1$ of $\mathcal{N}(\mu,\Sigma)$ conditioned on $x_{-1,2}^{(2)}$, discretize all these values, and then set the 2nd coordinate to be $\lfloor x_2^{(2)} \rceil  = n - \sum_{\ell \geq 3} \lfloor x_\ell^{(2)} \rceil - \lfloor x_1^{(2)} \rceil$. 

We couple the two sampling processes by letting $x_{-1,2}^{(1)} = x_{-1,2}^{(2)} := x_{-1,2}$.
With this in mind, we note that $x_{1}^{(1)} + x_{2}^{(1)} = x_{1}^{(2)} + x_{2}^{(2)} = n - \sum_{\ell \geq 3} x_\ell := \hat n$, where $x_1^{(1)}$ and $x_2^{(2)}$ are the ``unrounded'' versions of these coordinates.
$x_1^{(1)}$ is distributed independently and identically to $x_1^{(2)}$, and similarly for $x_2^{(2)}$ and $x_2^{(1)}$. 
We also define $n'$ to be $n - \sum_{\ell \geq 3} \lfloor x_\ell \rceil$. 
Ignoring the dimensions besides $1$ and $2$ (since they are coupled to be identical), the total variation distance between $Y_1$ and $Y_2$ is equal to the distance between $(n' - \lfloor x_2^{(1)} \rceil, \lfloor x_2^{(1)} \rceil)$ and $(\lfloor x_1^{(2)} \rceil, n' - \lfloor x_1^{(2)} \rceil)$.
By Lemma \ref{lem:DPI}, this is at most the total variation distance between $n' - \lfloor x_2^{(1)} \rceil$ and $\lfloor x_1^{(2)} \rceil$. 
Therefore, it suffices to upper bound the total variation distance between  $n' -  \lfloor x_2^{(1)} \rceil$ and $\lfloor \hat n - x_2^{(2)} \rceil$.
Since $n'$ is fixed, this is equal to the total variation distance between $\lfloor x_2^{(1)} \rceil$ and $\lfloor x_2^{(2)} + z \rceil$, where $z$ is some constant between $0$ and $k-1$. 
Again using Lemma \ref{lem:DPI}, this is upper bounded by the distance between $x_2^{(1)}$ and $x_2^{(2)} + z$.
Proposition \ref{prop:ddodtv} bounds this by $\frac{z}{2\s} \leq \frac{k}{2\s}$, as desired.
\end{proof}

As shown in Lemma \ref{lem:merge}, merging two Gaussians is cheap, assuming the minimum eigenvalues of their covariance matrices are sufficiently large.
The following lemma shows that this value stays large throughout the sequence of swap-and-merge operations.
\swapvar*

\begin{proof}
We need to prove that for all vectors $y \in \mathbb{R}^{S}$ such that $y_j = 0$ and $\|y\|_2 = 1$, $y^T \S y \ge \frac m {2 k^3}$. 

We have that $y^T \S y = \sum_i y^T \Sigma^{(i)} y \ge \max_i y^T \Sigma^{(i)} y$ since all matrices are positive semidefinite. We now consider a coordinate $j'$ of $y$ with maximum absolute value which has weight at least $\frac 1 {\sqrt k}$. Since the covariance matrix $\S$ is the result of summing matrices with common coordinates (by property 3 in the lemma statement), there is a sequence of coordinates starting from $j'$ and ending with $j$ that has length at most $k$, such that any two consecutive coordinates belong to at least one of the sets $S^{(i)}$. Since $|y_{j'}| \ge \frac 1 {\sqrt k}$ while $y_j =0$, it means that there exists a pair $(a,b)$ of consecutive coordinates in the path such that $|y_a - y_b| \ge \frac 1 {k \sqrt k}$.

Consider $\S^{(i)}$ such that $a,b \in S^{(i)}$. 
Let $j^* \in S^{(i)}$ be the coordinate such that $\Sigma^{(i)}_{S^{(i)} \setminus \{j^* \}}$ has minimum eigenvalue at least $\l$. We have that:
$$y^T \S^{(i)} y = y_{S^{(i)}}^T \S^{(i)}_{S^{(i)}} y_{S^{(i)}} =(y_{S^{(i)}} - y_{j^*} \vec 1_{S^{(i)}})^T \S^{(i)}_{S^{(i)}} (y_{S^{(i)}} - y_{j^*} \vec 1_{S^{(i)}}) \ge \l \|y_{S^{(i)}} - y_{j^*} \vec 1_{S^{(i)}} \|_2^2$$
where the second equality follows by property 1 in the lemma statement and the last inequality follows since $\Sigma^{(i)}_{S^{(i)} \setminus \{j^* \}}$ has minimum eigenvalue at least $\l$. Moreover since $|y_a - y_b| \ge \frac 1 {k \sqrt k}$, we have that $\|y_{S^{(i)}} - y_{j^*} \vec 1_{S^{(i)}} \|_2^2 \ge (y_a - y_{j^*})^2 + (y_b - y_{j^*})^2 \ge \frac 1 {2 k^3}$ which completes the proof of the lemma.
\end{proof}

Finally, with these two lemmas in hand, we can conclude with the proof of Theorem \ref{thm:struct}.

\begin{prevproof}{Theorem}{thm:struct}
  First, we justify the structure of the approximation, and then show that it can be $\ve$-close with our choice of the parameters $c$ and $t$.
  We start by applying Lemma \ref{lem:round} to obtain a PMD $M^{\hat \p}$ such that $\hat \p(i,j) \not \in (0,c)$ for all $i,j$.
  Partition the component CRVs into $k$ sets $S_1, \dots, S_k$, where the $i$th CRV is placed in the $l$th set if $l = \arg\max_j \hat \p(i,j)$ (with ties broken lexicographically).
  Since index $l$ is the heaviest, every CRV $i$ in $S_l$ has $\hat \r(i,l) \geq \frac1k$.
  We convert the PMD induced by each $S_l$ to a GMD by dropping the $l$th column.
  Applying Lemma \ref{lem:singlestruct} to each set and summing the results from all sets gives us a sum of $k$ Gaussians with a structure preserving rounding and a $(tk^2,k)$-Poisson multinomial random vector.
  Now, we iteratively merge the $k$ Gaussians: while there exists a pair of Gaussians who overlap in some dimension $\ell$ (i.e., there exists a dimension $\ell$ such that both Gaussians are not deterministically $0$), we merge them.
  To do this, we adjust the structure preserving rounding of both of the Gaussians to have pivot position $\ell$ (justified by Lemma \ref{lem:swaptv}), and then combine them by replacing their sum with a single Gaussian with a structure preserving rounding (using Lemmas \ref{lem:swapvar} and \ref{lem:merge}).
  This new Gaussian will have the same pivot $\ell$, and a mean vector and a covariance matrix equal to the sum of the two components.
  We repeat until we are left with a set of Gaussians which do not overlap, and then combine them into a single Gaussian with a structure preserving rounding, where each of the (disjoint) Gaussians corresponds to a different block.
  We note that Lemmas \ref{lem:singlestruct} and \ref{lem:swapvar} justify the minimum eigenvalue of each block of the covariance.

  Now, we show that our choices of $c$ and $t$ make the resulting distribution be $\ve$-close to the original.
  Applying Lemma \ref{lem:round} introduces a cost of $O\left(c^{1/2}k^{5/2}\log^{1/2}\left(\frac1{ck}\right)\right)$ in our approximation.
  We apply Lemma \ref{lem:singlestruct} $k$ times (once to each set $S_l$), so the total cost introduced here is 
  $O\left(\frac{k^{19/6}\log^{2/3}t}{c^{1/6}t^{1/6}} + \frac{k^{5/2}}{c^{1/2}t^{1/2}}\right)$.
  Lemma \ref{lem:swaptv} shows that each pivot swap costs $\frac{k}{2\s}$ in total variation distance.
  Lemma \ref{lem:swapvar} combined with \ref{lem:singlestruct} imply that $\s^2 \geq \frac{ct}{2k^4}$, and there are at most $2k$ pivot swaps, so this sequence of swaps costs at most $\frac{2k^4}{\sqrt{ct}}$.
  Similarly, by Lemma \ref{lem:merge}, each our (at most) $k$ merges costs $\frac{k}{2\s} \leq \frac{k^4}{\sqrt{ct}}$.
  Therefore, the total variation distance introduced in this entire sequence of operations is
  $$O\left(c^{1/2}k^{5/2}\log^{1/2}\left(\frac1{ck}\right) + \frac{k^{19/6}\log^{2/3}t}{c^{1/6}t^{1/6}} + \frac{k^{5/2}}{c^{1/2}t^{1/2}} + \frac{k^4}{\sqrt{ct}}\right).$$
  Recalling our choice of parameters, $c = \left(\frac{\ve^2}{k^5}\right)^{1+\d_c}, t = \left(\frac{k^{19}}{c\ve^6}\right)^{1 + \d_t}$ for $\d_c, \d_t > 0$, this results in a total variation distance which is $O(\ve)$.
\end{prevproof}



\section{Details from Section \ref{sec:coverpmd}}
\label{sec:appendixcover}

\subsection{A Direct Cover}
\label{sec:naivecover}
In this section, we present a direct cover of the class, following from the structural result of Theorem \ref{thm:struct}.
At a high level, we grid over the $O(k^2)$ parameters of the Gaussian component with granularity $\poly(\ve/k)/n$, and the $\poly(k/\ve)$ parameters of the $(tk^2,k)$-PMD with granularity $\poly(\ve/k)$, resulting in a cover of the claimed size.

\begin{prevproof}{Lemma}{lem:nonsparsecover}
  Our strategy will be as follows:
  Theorem \ref{thm:struct} implies that the original distribution is $O(\ve)$ close to a particular class of distributions.
  We generate an $O(\ve)$-cover for this generated class.
  By triangle inequality, this is an $O(\ve)$-cover for $(n,k)$-PMDs.
  In order to generate a cover, we will use a technique known as ``gridding''.
  We will generate a set of values for each parameter, and take the Cartesian product of these sets.
  Our guarantee is that the resulting set will contain at least one set of parameters defining a distribution which is $O(\ve)$-close to the PMD.

  First, observe that we can naively grid over the set of $(tk^2, k)$-PMDs.
  We note that if two CRVs have parameters which are within $\pm \frac{\ve}{k}$ of each other, then their total variation distance is at most $\ve$.
  Similarly, by triangle inequality, two PMDs of size $k^2t$ and dimension $k$ with parameters within $\pm \frac{\ve}{k^3t}$ of each other have a total variation distance at most $\ve$.
  By taking an additive grid of granularity $\frac{\ve}{k^3t}$ over all $k^2t$ parameters, we can generate an $O(\ve)$-cover for PMDs of size $k^2t$ and dimension $k$ with $O\left(\frac{k^3t}{\ve}\right)^{k^2t}$ candidates.

  Next, we wish to cover the Gaussian component.
  For a block, we will use $\m_i$ and $\S_i$ to refer to the mean and covariance, $n_i$ to the sum of the means within the block, and $S_i$ to refer to the set of coordinates.
  It will actually be more convenient to think of $\S_i$ in terms of a Cholesky decomposition $L_iL_i^T$\footnote{Recall that the Cholesky decomposition implies that $L_i$ will be lower triangular.}, which is guaranteed to exist since $\S_i$ is symmetric and positive semidefinite.
  We describe how to generate a $O\left(\frac{\ve}{k}\right)$-cover for a single block.
  We will prove that the underlying (continuous) Gaussians are $O\left(\frac{\ve}{k}\right)$ close, the closeness of the corresponding discretized versions follows by Lemma \ref{lem:DPI}.
  By taking the Cartesian product of the cover for each of the blocks and applying the triangle inequality, we generate a $O(\ve)$-cover for the overall Gaussian at the cost of a factor of $k$ in the exponent of the cover size.

  First, we examine the size parameter $n_i$.
  Since the size parameter is an integer between $0$ and $n$, we can simply try them all, giving us a factor of $n$ in the size of our cover.

  Covering the mean and covariance matrix takes a bit more care.
  We use Proposition \ref{prop:gaussapprox} to analyze the error incurred by inaccurate guesses for these parameters.
  We let $\mathcal{N}_1$ be the Gaussian corresponding to a single block of our Gaussian, and we will construct a $\mathcal{N}_2$ which is close to it.
  By Theorem $\ref{thm:struct}$, we know that $\s^2 \geq \frac{tc}{2k^4}$.

  We examine the first term of the bound in Proposition \ref{prop:gaussapprox}.
  If $\| \m_1 - \m_2\|_2 \leq \frac{\ve c \sqrt{t}}{k} \coloneqq \b$, then this term is $O\left(\frac{\ve}{k}\right)$.
  Note that $\m_1 \in [0,n]^{|S_i|-1}$, since each coordinate is the sum of at most $n$ parameters which are at most $1$.
  We can create a $\b$-cover of size $O\left(\frac{n\sqrt{k}}{\b}\right)^{O(k)}$ for this space, with respect to the $\ell_2$ distance.
  To see this, consider covering the space with $(|S_i|-1)$-cubes of side length $\frac{\b}{\sqrt{|S_i|-1}}$.
  Any two points within the same cube are at $\ell_2$ distance at most $\b$.
  Taking a single vertex from the $(|S_i|-1)$-cube, our cover is defined by taking the corresponding vertex from all the cubes.
  The volume of $[0,n]^{|S_i|-1}$ is $n^{|S_i|-1}$, and the volume of each cube is $\left(\frac{\b}{\sqrt{|S_i|-1}}\right)^{|S_i|-1}$, so the total number of points needed is $\left(\frac{n\sqrt{|S_i| -1}}{\b}\right)^{|S_i|-1}$.
  Substituting in the value of $\b$ shows that a set of size $O\left(\frac{nk\sqrt{|S_i| - 1}}{\ve c \sqrt{t}}\right)^{|S_i| -1} = O\left(\frac{nk^{3/2}}{\ve c \sqrt{t}}\right)^{|S_i|-1}$ suffices to cover the mean of the Gaussian to a sufficient accuracy.

  Next, we examine the second term in Proposition \ref{prop:gaussapprox}.
  Taking $\a \leq \frac{\ve ct}{2k^6}$ sets this term to be $O\left(\frac{\ve}{k}\right)$.
  However, we can not naively grid over the matrices, since the covariance matrix is required to be PSD. 
  Therefore, we instead grid over entries of a Cholesky decomposition.
  Observe that the diagonal entries of the true covariance matrix are equal to the $\ell_2$ norms of the rows of the true Cholesky decomposition.
  Since the maximum entry in the true covariance matrix is at most $n$, this implies that the magnitude of the maximum entry of the true Cholesky decomposition is at most $\sqrt{n}$.
  If we grid over the entries of the Cholesky decomposition with granularity $\gamma$, there will exist a candidate where all entries are within $\pm \gamma$ of the true entries.
  Using the bound of $\sqrt{n}$ on the maximum element and $|S_i| - 1$ as the dimension of the matrix, this will imply that the entries of the resulting covariance matrix are within $\pm (2\gamma\sqrt{n} + \gamma^2)(|S_i| - 1) \leq 4\gamma|S_i|\sqrt{n}$ of the true entries. 
  Since we want this value to be upper bounded by $\alpha$, it gives that $\gamma \leq \frac{\ve c t}{8k^6 |S_i|\sqrt{n}}$.
  Combining this gridding granularity with the fact that there are at most $\frac{|S_i|^2 - |S_i|}{2}$ non-zero entries in the Cholesky decomposition, which are in the range $[-\sqrt{n}, \sqrt{n}]$,
  this implies a cover of size at most $O\left(\frac{nk^7}{\ve c t}\right)^{\frac12 |S_i|^2 - \frac12 |S_i|}$.

  Combining the gridding for the size, mean, and covariance, a $O\left(\frac{\ve}{k}\right)$-cover for one block is of size
  $$n O\left(\frac{nk^{3/2}}{\ve c \sqrt{t}}\right)^{|S_i|-1} O\left(\frac{nk^7}{\ve c t}\right)^{\frac12 |S_i|^2 - \frac12 |S_i|} = n^{\frac12 |S_i|^2 + \frac12 |S_i|}\left(\frac{k}{\ve c t}\right)^{O(|S_i|^2)}.$$
  Taking the Cartesian product over all the blocks of the Gaussian and noting this function is convex in the values of $\{|S_i|\}$, we cover the entire Gaussian with a set of size at most
  $$n^{\frac12 k^2 + \frac12 k}\left(\frac{k}{\ve c t}\right)^{O(k^2)}.$$

  Combining the cover for the Gaussian component and the $(tk^2, k)$-PMD gives us a cover of size 
  $$n^{\frac12 k^2 + \frac12 k}\left(\frac{k}{\ve c t}\right)^{O(k^2)}\left(\frac{k^3t}{\ve}\right)^{k^2t}.$$
  Substituting in the values of $c$ and $t$ gives us a cover of size
  $$n^{\frac12 k^2 + \frac12 k}\left(\frac{k}{\ve}\right)^{O\left(\frac{k^{26 + \d_1}}{\ve^{8 + \d_2}}\right)}.$$
  for constants $\d_1, \d_2 > 0$, which satisfies the statement of the theorem.
\end{prevproof}

\subsection{A Sparser Cover}
\label{sec:sparsecover}
In the previous section, we chose a naive gridding for the sparse component, resulting in a cover size which is exponential in $\poly(k/\ve)$.
In this section, we present the cover described by Lemma~\ref{lem:sparsecover}, which is of size exponential in $k^{5k} \log^{k+2} (1/\ve)$.
We use a moment matching technique similar to Roos \cite{Roos02}.
In this work, Roos showed that any generalized multinomial distribution can be written as a weighted summation of derivatives of a simple multinomial distribution. 
To describe his theorem, we define $V_k(n) = \{v \in \mathbb{Z}^k: v_i \ge 0 \wedge \sum_i v_i \le n\}$.

\begin{lemma}[Theorem 1 in \cite{Roos02}]
\label{lem:approximatorcoefficients}
For an arbitrary vector $\vec q$ with $|\vec q| \le 1$, the density of the generalized multinomial distribution $M^\r$ at any point $x$ can be expressed as:
$$\sum_{u \in V_{k}(n)} a_{u}(\vec q) \Delta^u \mathcal{M}(n-|u|,\vec q, x)$$
where $\mathcal{M}(n, \vec q, x)$ represents the density of the multinomial distribution with probabilities $\vec q$ at point $x$ and $a_{u}(\vec q)$ is the coefficient of the term $\prod_j z_j^{u_j}$ in the expansion of the polynomial:
$$\prod_{i=1}^n \left( 1 + \sum_{j=1}^k (\r(i,j)-q_j) z_j \right)$$
\end{lemma}

Roos also showed that considering fewer terms in the summation above provides a good approximation to the density of the original generalized multinomial distribution. We consider the approximator: 
$$m_{w,\vec q} (x) = \sum_{u \in V_{k}(w)} a_{u}(\vec q) \Delta^u \mathcal{M}(n-|u|,\vec q, x)$$

\begin{lemma}[Theorem 2 in \cite{Roos02}]
\label{lem:approximator}
$$||M^\r - m_{w,\vec q}||_1 \le \frac{ \alpha^{w+1} } { 1 - \alpha } \,\, \textrm{ for } \alpha < 1$$
where $$\alpha =  e^{1/2} \sum_{j=1}^k \sqrt{\frac{2 \sum_{i=1}^n (\r(i,j)-q_j)^2 + \left(\sum_{i=1}^n (\r(i,j)-q_j) \right)^2}{2 n q_0 q_j}}$$
\end{lemma}

We will use these results to produce a sparser cover. We will first show that for a particular class of generalized multinomial distributions, there exist good approximators.

\begin{lemma}
\label{lem:approximatorbound}
Consider a generalized multinomial distribution $M^{\r}$. If for all $j \in [k]$, it holds that  $|\max_i \r(i,j) - \min_i \r(i,j)| \le ( 4 e k^3 )^{-1}$ and moreover $\sum_{i=1}^n \r(i,0) \ge \frac n k$, then:
$$||M^{\r}- m_{w,\vec q}||_1 \le 2^{-w}$$
for the vector $\vec q$ with $q_j = \frac 1 n \sum_{i=1}^n \r(i,j)$
\end{lemma}
\begin{proof}
Since according to Lemma~\ref{lem:approximator} the $\ell_1$ distance of $M^{\r}$ to the approximator is at most $\frac{ \alpha^{w+1} } { 1 - \alpha }$, it suffices to show that $\alpha \le \frac 1 2$.

By our choice of $\vec q$ it holds that $\sum_{i=1}^n (\r(i,j)-q_j) = 0$. Therefore, we have that:
$$\alpha =  e^{1/2} \sum_{j=1}^k \sqrt{\frac{\sum_{i=1}^n (\r(i,j)-q_j)^2}{n q_0 q_j}} \le e^{1/2} \sum_{j=1}^k \sqrt{k} \sqrt{\frac{\sum_{i=1}^n (\r(i,j)-q_j)^2}{n q_j}}$$
since $q_0 \ge \frac 1 k$. Moreover, we have that $\frac{\sum_{i=1}^n (\r(i,j)-q_j)^2}{n q_j} \le |\max_i \r(i,j) - \min_i \r(i,j)| \le  ( 4 e k^3 )^{-1}$. Plugging this bound in the above expression for $\alpha$ gives the desired bound.
\end{proof}

We now show that if two PMDs have matching moments then their approximators are the same. This will allow us to compare the total variation between them by looking at their distance to the common approximator.

\begin{lemma}
\label{lem:sameapproximators}
Consider two generalized multinomial distributions $M^{\r}$, $M^{\r'}$ and their approximators $m_{w,\vec q}$ and $m'_{w,\vec q}$. If for all $u \in V_k(w)$ and $j \in [k]$:
$$\sum_{i=1}^n \prod_{j=1}^k \r(i,j)^{u_j} = \sum_{i=1}^n \prod_{j=1}^k \r'(i,j)^{u_j}$$
then $m_{w,\vec q} = m'_{w,\vec q}$
\end{lemma}

\begin{proof}
We first note that if the condition holds for all $u \in V_k(w)$, then it also holds that for all $u \in V_k(w)$: 
$$\sum_{i=1}^n \prod_{j=1}^k (\r(i,j) - q_j)^{u_j} = \sum_{i=1}^n \prod_{j=1}^k (\r'(i,j) - q_j)^{u_j}$$
This is because when expanding the product $\prod_{j=1}^k (\r(i,j) - q_j)^{u_j}$ and treating it as a polynomial in $q_j$, the coefficients in each term are a polynomial of degree at most $w$ in the $\r(i,j)$ and summing over all $i$ we get that the two sides are equal.

We now define $\bar \r(i,j) = \r(i,j) - q_j$ and note that according to Lemma~\ref{lem:approximatorcoefficients}, the coefficients of the approximator $m_{w,\vec q}$ are given by the expansion of the polynomial: $\prod_{i=1}^n \left( 1 + \sum_{j=1}^k \bar \r(i,j) z_j \right)$. We observe that for any given $u$, the coefficient $a_{u}(\vec q)$ of the term $\prod_j z_j^{u_j}$ is a degree $|u|$ polynomial in terms of $\bar \r(i,j)$ which, by the theory of multisymmetric polynomials, can be written entirely as a polynomial of the elementary multisymmetric polynomials, 
$\sum_{i=1}^n \prod_{j=1}^k \bar \r(i,j)^{v_i}$ for $v \in V_k(w)$. 
Since $M^{\r}$ and $M^{\r'}$ are equal in all those terms, it means that they have equal coefficients $a_{u}(\vec q)$ and thus their approximators are the same.
\end{proof}

Using those two lemmas, we can construct a cover for $(tk^2,k)$-PMD which has an exponentially better dependence on $1/\ve$.
We must cover at most $k^2 t$ CRVs, which we can assume each have probabilities that are multiples of $\frac {\ve} {k^3 t}$. 
By the previous section, this induces a cost of $O(\ve)$ in total variation distance.
To apply Lemma~\ref{lem:approximatorbound}, we will first partition the CRVs into $( 4 e k^3 )^k$ groups.
In particular, consider indexing the groups by $\vec v \in [4ek^3]^k$.
In group $v$, we include all CRVs with mean vector $p$ where $p_j \in \frac{1}{4ek^3}[v_j - 1, v_j]$ for all $j \in [k]$.
For the PMD induced by each group, we have the property $|\max_i \r(i,j) - \min_i \r(i,j)| \le ( 4 e k^3 )^{-1}$. 
We cover each such PMD separately by considering all possible different moment profiles that it can achieve. 
A moment profile for a PMD of size $n$ is a vector of $|V_k(w)|$ elements, where the entry of the profile indexed by $u \in V_k(w)$ is equal to $\sum_{i=1}^n \prod_{j=1}^k \r(i,j)^{u_j}$. 
By Lemma~\ref{lem:sameapproximators} if two PMDs have the same moment profiles they have the same approximator and thus by Lemma~\ref{lem:approximatorbound} and triangle inequality their total variation is at most $2^{-w+1}$.

We now count how many different moment profiles are possible to arise. For a given $u$, there are at most $k \frac {k^5 t^2} {\ve}$ different values when $|u| = 1$, $\left( k \frac {k^5 t^2} {\ve} \right)^2$ values for $|u| = 2$, and in general $\left( k \frac {k^5 t^2} {\ve} \right)^i$ values when $|u|=i$. Since there are at most $i^{k-1}$ vectors with $|u|=i$, there are at most 
$$\prod_{i=1}^w \left( \frac {k^6 t^2} {\ve} \right)^{ i^k } = \left( \frac {k} {\ve} \right)^{O( w^{k+1} )}$$
different moment profiles. By picking $w = k \log( \frac {4 e k^3} {\ve} )$, we get small enough error so that union bounding over all $( 4 e k^3 )^k$ different groups will still give an $\ve$ error. 
This means that by considering only a single PMD for each moment profile in each of the $( 4 e k^3 )^k$ groups, we can create an $\ve$-cover of size $\left( \frac {k} {\ve} \right)^{O( ( 4 e k^3 )^k k^{k+1} \log^{k+1} ( \frac {4 e k^3} {\ve} ) )} = 2^{O(  k^{5 k} \log^{k+2}( \frac 1 \ve ) )}$, concluding the proof of Lemma \ref{lem:sparsecover}.

\section{Details from Section \ref{sec:learningpmd} }
\label{app:learning}


\subsection{Estimating the mean and covariance of a PMD}
\label{sec:momentestimate}
We will prove an analogue of Lemma 6 in \cite{DaskalakisDS12b}, i.e., that we can accurately estimate the mean and covariance of a PMD with a small number of samples.
First, we will show that we can get accurate estimates of the moments in any particular direction we desire.
Then, taking the union bound over $k^2$ directions, we show that our estimate is accurate for all directions simultaneously.

\begin{lemma}\label{lem:1destimate}
  For any vector $y$, given sample access to a $(n,k)$-PMD $X$ with mean $\m$ and covariance matrix $\S$, there exists an algorithm which can produce estimates $\hat \mu$ and $\hat \S$ such that with probability at least $9/10$:
  $$|y^T(\hat \mu - \mu)| \leq \ve \sqrt{y^T \S y} ~~ \text{and} ~~ |y^T(\hat \S - \S)y| \leq \ve y^T \S y \sqrt{1 + \frac{y^T y}{y^T \S y}} $$

  The sample and time complexity are $O(1/\ve^2)$.
\end{lemma}
\begin{proof}
  We start with the estimate $\hat \m$.
  Let $Z_1, \dots, Z_m$ be independent samples from $X$, and let $\hat \mu = \frac{1}{m}\sum_i Z_i$. Then 
      $$E[y^T \hat \m] = y^T\m ~~ \text{and} ~~ Var[y^T \hat \m] = \frac{1}{m} Var[y^TX] = \frac{y^T \S y}{m}.$$
      Then by Chebyshev's inequality,
      $$\Pr[|y^T(\hat \m - \m)| \geq t\sqrt{y^T \S y}/\sqrt{m}] \leq \frac{1}{t^2}.$$
      Choosing $t = \sqrt{10}$ and $m = \lceil 10/\ve^2 \rceil$, the above imply that $|y^T(\hat \m - \m)| \leq \ve \sqrt{y^T \S y}$ with probability at least $9/10$.

    Next, we describe $\hat \S$.
    Let $Z_1, \dots, Z_m$ be independent samples from $X$, and let the empirical estimator for the covariance be $\hat \S = \frac{1}{m-1}\sum_i (Z_i - \frac{1}{m}\sum_i Z_i) (Z_i - \frac{1}{m}\sum_i Z_i)^T$. 
      Then it can be shown that \cite{Johnson11}:
      $$E[y^T \hat \S y] = y^T \S y ~~ \text{and} ~~ Var[y^T \hat \S y] = (y^T \S y)^2 \left(\frac{2}{m-1} + \frac{\k_y}{m}\right) $$
      where $\k_y$ is the excess kurtosis of the distribution of $X$ with respect to the vector $y$ (i.e., $\k_y = \frac{E[(y^T (X - \m))^4]}{(y^T\S y)^2} -3$).

      It can be shown that:
      $$\k_y \leq \sum_{i=1}^n \frac{E[(y^T(X_i - \m))^4]}{(y^T\S y)^2},$$
      where $X_i$ is the $i$th CRV in the PMD. 
      We note that $y^T(X_i - \m) \leq 2\|y\|_2$.
      This is because $\|X_i\|_2 = 1$, $\|\m\|_1 = 1$ and $\|\m\|_2 \leq \|\m\|_1$.
      Therefore $(y^T(X_i - \m))^4 \leq 4\|y\|_2^2(y^T(X_i - \m))^2$, and thus
      $$\k_y \leq \sum_{i=1}^n \frac{4\|y\|_2^2E[(y^T(X_i - \m))^2]}{(y^T\S y)^2} = 4\frac{\|y\|_2^2y^T\S y}{(y^T\S y)^2} = \frac{4y^T y}{y^T \S y} $$
      Therefore, $Var[y^T \hat \S y] \leq  (y^T \S y)^2 \left(\frac{2}{m-1} + \frac{4y^T y}{m (y^T \S y)}\right) \leq
      \frac{4(y^T \S y)^2}{m} \left(1 + \frac{y^T y}{y^T \S y}\right)$.
      Again using Chebyshev's inequality, 
      $$\Pr\left[|y^T(\hat \S - \S)y| \geq t\frac{2y^T \S y}{\sqrt{m}}\sqrt{1 + \frac{y^T y}{y^T \S y}}\right] \leq \frac{1}{t^2}.$$
      Taking $t = \sqrt{10}$ and $m = \lceil 40/\ve^2 \rceil$, the above imply that $|y^T (\hat \S - \S) y| \leq \ve y^T \S y \sqrt{1 + \frac{y^T y}{y^T \S y}}$ with probability at least $9/10$.
\end{proof}

\begin{lemma}
  \label{lem:1dtokdestimate}
Let $\S,\hat \S \in \mathbb{R}^{k \times k}$ be two symmetric, positive semi-definite matrices, and let $(\l_1,v_1), \dots, (\l_k,v_k)$ be the eigenvalue-eigenvector pairs of $\S$.
Suppose that
\begin{itemize}
\item For all $i \in [k]$, $\Big|\Big(\frac{v_i}{\sqrt{\l}_i}\Big)^T \Big(\hat \S - \S\Big) \Big(\frac{v_i}{\sqrt{\l}_i}\Big)\Big|  \leq \ve $,
\item For all $i, j \in [k]$, $\Big|\Big(\frac{v_i}{\sqrt{\l}_i} + \frac{v_j}{\sqrt{\l}_j}\Big)^T \Big(\hat \S - \S\Big) \Big(\frac{v_i}{\sqrt{\l}_i} + \frac{v_j}{\sqrt{\l}_j}\Big)\Big| \leq 4\ve$.
\end{itemize}
Then for all $y \in \mathbb{R}^k$, $\Big|y^T \Big(\hat \S - \S\Big) y\Big| \leq 3k\ve y^T \S y$.
\end{lemma}
\begin{proof}
  Without loss of generality, we can focus on the case $\S = I$, with eigenvalue-eigenvector pairs $(1,e_j)$ for all $j \in [k]$.
  To see this, write $\S$ as its eigendecomposition $Q \Lambda Q^T$, and replace $y$ with $Q \Lambda^{-1/2} x$, which will place the matrix $\S$ in ``isotropic position.''
  
  We now have the guarantees
  \begin{itemize}
    \item For all $i \in [k]$, $|e_i^T (\hat \S - I) e_i|  \leq \ve $,
    \item For all $i, j \in [k]$, $|(e_i + e_j)^T (\hat \S - I) (e_i + e_j)| \leq 4\ve$,
  \end{itemize}
  and we wish to show $|y^T (\hat \S - I) y| \leq 3k\ve \|y\|_2^2$ for all $y \in \mathbb{R}^k$.
  
  We need the following proposition: 
  \begin{proposition}
    For any vector $x \in \mathbb{R}^k$ and matrix $A \in \mathbb{R}^{k \times k}$
    $$x^T A x = \frac12 \sum_{i \neq j} x_i x_j (e_i + e_j)^T A (e_i + e_j) + 2 \sum_i x_i^2 e_i^T A e_i - \Big(\sum_i x_i\Big) \sum_i x_i e_i^T A e_i,  $$
    where $e_i$ is the $i$th standard basis vector.
  \end{proposition}
  \begin{proof}
    Observe that
    \begin{align*}
      \frac12\sum_{i \neq j} x_i x_j (e_i + e_j)^T A (e_i + e_j) &= \frac12 \sum_{i \neq j} x_i x_j(A_{ii} + A_{jj} + A_{ij} + A_{ji}) \\
                                                                 &= \sum_{i \neq j} x_i x_j A_{ij} + \Big(\sum_i \Big(\sum_{j \neq i} x_j\Big) x_i A_{ii}\Big).
    \end{align*}
    Adding the $ 2 \sum_i x_i^2 e_i^T A e_i$ term gives us 
    $$\sum_{i,j} x_i x_j A_{ij} + \Big(\sum_i x_i\Big) \Big(\sum_i x_i A_{ii}\Big) = \sum_{i,j} x_i x_j A_{ij} + \Big(\sum_i x_i\Big) \Big(\sum_i x_i e_i^T A e_i\Big)$$
    Subtracting the final term gives the desired result.
  \end{proof}
  We apply this to $|y^T (\hat \S - I) y|$, giving
  $$\Big|y^T (\hat \S - I) y\Big| = \Big| \frac12 \sum_{i \neq j} y_i y_j (e_i + e_j)^T (\hat \S - I) (e_i + e_j) + 2 \sum_i y_i^2 e_i^T (\hat \S - I) e_i - \Big(\sum_i y_i\Big) \sum_i y_i e_i^T\Big(\hat \S - I\Big) e_i\Big|.$$

  Using the guarantees in the lemma statement,
  \begin{align*} 
    &\Big|\frac12 \sum_{i \neq j} y_i y_j (e_i + e_j)^T (\hat \S - I) (e_i + e_j) + 2 \sum_i y_i^2 e_i^T (\hat \S - I) e_i - \Big(\sum_i y_i\Big) \sum_i y_i e_i^T\Big(\hat \S - I\Big) e_i\Big| \\
    &\leq \ve \Big(2 \sum_{i \neq j} |y_i| |y_j| + 2 \sum_i y_i^2  + \Big(\sum_i y_i\Big)^2\Big) \\
    &= \ve\Big(2  \sum_{i,j} |y_i| |y_j| + \|y\|_1^2\Big) \\
    &= 3\ve  \|y\|_1^2 \leq 3k\ve \|y\|_2^2
  \end{align*}
  where the final inequality is Cauchy-Schwarz. 
\end{proof}

\allestimate*
\begin{proof}
  The proof will follow by applying Lemma \ref{lem:1destimate} to $k^2$ carefully chosen vectors simultaneously using the union bound.
  Using the resulting guarantees, we show that the same estimates hold for any direction, at a cost of rescaling $\ve$ by a factor of $k$.

  Let $S$ be the set of $k^2$ vectors $\{v_i\}$ and $\Big\{\frac{v_i}{\sqrt{\l}_i} + \frac{v_j}{\sqrt{\l}_j}\Big\}$ for all $(i,j) \in [k] \times [k]$, where $\{(\l_i,v_i)\}$ are the (unknown) eigenvalue-eigenvector pairs of $\S$.
  From $O(k^4/\ve^2)$ samples, with probability $9/10$, we can obtain estimators $\hat \m$ and $\hat \S$ such that
  $$|y^T(\hat \mu - \mu)| \leq \frac{\ve}{k} \sqrt{y^T \S y} ~~ \text{and} ~~ |y^T(\hat \S - \S)y| \leq \frac{\ve}{3k}  y^T \S y.$$
  This follows by Lemma \ref{lem:1destimate}, the eigenvalue condition on $\S$, and an application of the union bound.  
  
  We first prove that the mean estimator $\hat \m$ is accurate. 
  Consider an arbitrary vector $y$, which can be decomposed into a linear composition of the eigenvectors $y = \sum_i \alpha_i v_i$.
  \begin{align*}
    |y^T(\hat \m - \m)| &= |\sum_i \alpha_i v_i^T (\hat \m - \m) | \leq \sum_i |\alpha_i| |v_i^T (\hat \m - \m)| \leq \frac{\ve}{k} \sum_i |\alpha_i|\sqrt{\l_i} \leq \frac{\ve}{k} \sqrt{k} \sqrt{ \sum_i \alpha_i^2{\l_i} }
  \end{align*}
  where the last inequality is Cauchy-Schwarz. 
  Since $ \sum_i \alpha_i^2{\l_i} = y^T \S y$, this proves the desired accuracy bound for the mean's estimator.
  
  The accuracy of $\hat \S$ follows from an application of Lemma \ref{lem:1dtokdestimate}.
\end{proof}

\subsection{Rounding preserves the mean and covariance}
\label{sec:roundingmoments}

In order to convert our estimate of the covariance matrix for the PMD to an estimate of the covariance matrix for the Gaussian component, we first need to understand how much the rounding step affected the covariance matrix. We will use the fact that the unrounded GMD we are sampling from and the rounded GMD we want to estimate are $\ve$-close in total variation and show the following lemma:

\begin{lemma}
\label{lem:gmd-rounding-distance}
  Suppose there exist two $\ve$-close $(n,k)$-GMDs with covariance matrices $\S_1$ and $\S_2$, where the minimum eigenvalue of $\S_1$ is at least $1/\ve^{3}$. 
  Then for any vector $y$, $|y^T (\S_1 - \S_2) y| \leq 9 \ve y^T \S_1 y$.
\end{lemma}

\begin{proof}
Since the variance of the GMD with covariance matrix $\S_1$ is at least $1/\ve^{3}$ when projecting to direction $y$, we can apply the Berry-Esseen theorem (Proposition~\ref{prop:berryesseen}) to show that it is close in Kolmogorov distance to a Gaussian with the same mean and variance $ y^T \S_1 y$. 
To do this, we first re-center the GMD by subtracting the mean from each summand and projecting in direction $y$ with $\|y\|_2 = 1$.
This gives us a sum of $n$ independent random variables that lie in $[-\sqrt{2},\sqrt{2}]$.
This implies that $\r_i \le \sqrt{2}\s_i^2$ and Proposition~\ref{prop:berryesseen} gives the Kolmogorov distance induced to be:
$$ \frac{1}{( \sum_{i=1}^n \s_i^2 )^{1/2}} = \frac{1}{( y^T \S_1 y )^{1/2}}= \ve^{3/2} \le \ve $$

We will now show that the variance of the second GMD along direction $y$ needs to also be at least $1/\ve^{2}$, in order for the two GMDs to have total variation distance less than $\ve$.
We assume that this is not the case, for the sake of contradiction.

Consider the random variable $Y$ that is distributed according to the second GMD in direction $y$. By Chebyshev's inequality, we have that:
$\Pr[ |Y-E[Y]| > \sqrt{ 1/\ve^{3} } ] \le \ve$. 
However, the first GMD has $\Omega(1)$ probability mass distributed outside the interval one standard deviation from its mean, since it is well approximated in Kolmogorov distance (and thus, by Fact~\ref{fct:drelation}, in total variation distance) by a Gaussian.
Therefore, the two GMDs are $\Omega(1)$-far, which is a contradiction.

Now, since the second GMD has minimum variance at least $1/\ve^{2}$, we can also approximate it by a Gaussian as before using the Berry-Esseen Bound, losing $\ve$ in Kolmogorov distance.
Proposition~\ref{prop:close-gaussians} then implies that in order for the total variation distance between the two to be at most $3 \ve$, we must have that $|y^T (\S_1 - \S_2) y| \leq 9 \ve y^T \S_1 y$.
\end{proof}

\begin{proposition}
\label{prop:close-gaussians}
  For two single dimensional Gaussians $\mathcal{N}_1 = \mathcal{N}(\m_1,\s_1^2)$, $\mathcal{N}_2 = \mathcal{N}(\m_2, \s_2^2)$ such that $\frac{\s_1}{\s_2} \not \in (1 - \ve, 1 + \ve)$, it holds that
  $$\dk(\mathcal{N}(\m_1,\s_1^2), \mathcal{N}(\m_2, \s_2^2)) \ge \frac {\ve} {3}.$$ 
\end{proposition}
\begin{proof}
  Without loss of generality, suppose $\m_1 \leq \m_2$ and $\s_1 \leq \s_2$.
  Consider the point $x = \m_1 + \sqrt{2}\s_1$. 
  At this point, the CDF of the first Gaussian equal to $\frac{1}{2}(1 + \erf(1))$. 
  Similarly, the CDF of the second Gaussian is at most $\frac{1}{2}(1 + \erf(\frac{\s_1}{\s_2})) \leq \frac{1}{2}(1 + \erf(1 - \ve))$.
  Therefore, $\dk(\mathcal{N}_1, \mathcal{N}_2) \geq \frac{\erf(1) - \erf(1 - \ve)}{2} \ge \frac \ve 3$ where the last inequality holds for all $\ve \in (0,1)$.
\end{proof}

Applying Lemma~\ref{lem:gmd-rounding-distance} implies that our estimate for the PMD's covariance matrix is also a good estimate of the covariance matrix after applying the rounding procedure described in Section~\ref{sec:rounding}.
Moreover, the mean is preserved almost exactly since, by construction, there is a small additive error of $c$ in each coordinate.
Since the minimum eigenvalue of the PMD's covariance matrix is at least $1$, this additive error is negligible.

\subsection{Converting moment estimates from the PMD to the Gaussian} \label{sec:Gaussian}

In the previous sections, we showed how to estimate the moments of the rounded PMD. 
However, we can not use these estimates to obtain the moments of the Gaussian component of the structure directly.
The problem is that since the rounded $(n,k)$-Poisson multinomial random vector might be the sum of a Gaussian and a $(tk^2,k)$-Poisson multinomial random vector, the empirical mean and covariance of the samples might be very different than the mean and covariance of the Gaussian component we want to estimate. 
In this section, we show how to convert our estimates to accurately describe the Gaussian component by appropriately guessing the error induced by the non-Gaussian component.

%
%
%
%
%

Let $(\mu,\S), (\mu_G,\S_G), (\mu_S, \S_S)$ be the means and covariance matrices of the (rounded) $(n,k)$-PMD, of the Gaussian component, and of the $(tk^2,k)$-PMD respectively and $(\hat \mu, \hat \S)$ be the empirical mean and covariance matrix we estimated in the previous section. It holds that $\mu = \mu_G + \mu_S$ and $\S = \S_G + \S_S$.

By Lemma~\ref{lem:allestimate} and Lemma~\ref{lem:gmd-rounding-distance}, after taking $O(k^4/\ve^2)$ samples, with high probability, we have that for all vectors $y$, $|y^T (\hat \mu - \mu)| \le \ve \sqrt{ y^T \S y}$ and $|y^T(\hat \S - \S)y| \leq \ve y^T \S y $. 
We show how to correct our estimate $(\hat \mu, \hat \S)$.
In particular, we will generate a set of candidates which contains an estimate $(\hat \mu_G, \hat \S_G)$ such that for all vectors $y$, $|y^T (\hat \mu_G - \mu_G)| \le \ve \sqrt{ y^T \S_G y}$ and $|y^T(\hat \S_G - \S_G)y| \leq \ve y^T \S_G y$. 
We do this without any additional samples, by carefully gridding around the estimated mean and covariance. 

To achieve the guarantee for the covariance matrix, we compute a sparse cover of the space of all PSD matrices around $\hat \S$.

\begin{definition}
  Let S be a set of symmetric $k \times k$ PSD matrices. An \emph{$\ve$-cover} of the set $S$, denoted by $S_{\ve}$, is a set of PSD matrices such that for any matrix $A \in S$, there exists a matrix $B \in S_{\ve}$ such that for all vectors $y$: $|y^T (A-B) y| \le \ve y^T A y$.
\end{definition}

Using the fact that $|y^T(\hat \S - \S)y| \leq \ve y^T \S y$ and $|y^T(\S - \S_G)y| = |y^T \S_S y| \le m y^T y$, we know that
$|y^T(\hat \S - \S_G)y| \leq \frac \ve {1-\ve} y^T \hat \S y+ m y^T y \le 2 \ve y^T \hat \S y+ m y^T y$.
This means that in order to get an estimate $\hat \S_G$ such that for all directions $y$, $|y^T(\hat \S_G - \S_G)y| \leq \ve y^T \S_G y$, it suffices to consider an $\ve$-cover of the PSD matrices $A$ that satisfy the property $|y^T(\hat \S - A)y| \leq 2 \ve y^T \hat \S y+ m y^T y$ for all vectors $y$. The following lemma gives an efficient construction of the cover and bounds its size.

\psdcover*

\begin{proof}

To construct the cover, we will make use of the eigenvalues and eigenvectors of the matrix $A$.
We first show that for any matrix $B \in S$, its eigenvalues are close to the eigenvalues of $A$. 

\begin{proposition} \label{lem:courant}
Let $A,B$ be two symmetric $k \times k$ PSD matrices such that for all vectors $y$ with $\|y\|=1$, $|y^T(A - B)y| \le \ve_1 y^T A y + \ve_2$ for some constants $\ve_1, \ve_2 > 0$. Then for the eigenvalues $\lambda^A_1 \le ... \le \lambda^A_k$ of $A$, and the eigenvalues $\lambda^B_1 \le ... \le \lambda^B_k$ of $B$, it holds that:
$$|\lambda^A_i - \lambda^B_i| \le \ve_1 \lambda^A_i +\ve_2$$
\end{proposition}

\begin{proof}
From Courant's minimax principle, we have that the $i$-th eigenvalue of $A$ is equal to:
$$\lambda^A_i=\max\limits_C\min\limits_{\binom{\| x\| =1}{Cx=0}}  x^T Ax$$
where $C$ is an $(i-1) \times k$ matrix. For the matrix $B$, we have that
$$\lambda^B_i=\max\limits_C\min\limits_{\binom{\| x\| =1}{Cx=0}}  x^T Bx \le \max\limits_C\min\limits_{\binom{\| x\| =1}{Cx=0}}  (1+\ve_1) x^T Ax + \ve_2 =  (1+\ve_1) \lambda_i^A + \ve_2 $$
Similarly, we have that $\lambda^B_i \ge (1-\ve_1) \lambda_i^A - \ve_2 $, so the result follows.
\end{proof}

This means that by computing the eigenvalues $\mu_1 \le .. \le \mu_k$ of $A$ and then guessing $8 \ve_2$ possible values to subtract in the range $[-\ve_2,\ve_2]$ with accuracy $1/4$, we can get estimates of the eigenvalues $\lambda_1, ... , \lambda_k$ of $B$ within a multiplicative factor of $1 \pm 1/2$. This is true because the minimum eigenvalue of $B $ is at least $1$. We can improve our estimates to a better multiplicative factor $1 \pm \ve$ by gridding multiplicatively around each eigenvalue. This requires another $\log_{1+\ve}\left(\frac{1+1/2}{1-1/2}\right) = O(1/\ve)$ guesses per eigenvalue. So in total, we require $\left( \frac {1 + \ve_2} {\ve} \right)^{O(k)}$ guesses for obtaining accurate estimates $\lambda'_1, ... , \lambda'_k$ of the eigenvalues of $B$.

Once we know (approximately) the eigenvalues of $B$, we will try to guess also its eigenvectors $v_1,...,v_k$. We will do this by performing a careful gridding around the eigenvectors of $A$ which we can assume, without loss of generality (by rotating), to be the standard basis vectors $e_1, e_2, ...,e_k$. So for each eigenvector $v_z$ of $B$, we will try to approximate it by guessing its projections to the eigenvectors of $A$. 

We now bound the projections of eigenvectors of $A$ to eigenvectors of $B$. Since we know that  $e^T_i B e_i \le (1+\ve_1) e^T_i A e_i + \ve_2$, we get that $\sum_z \lambda_z (v_z e_i)^2 \le (1+\ve_1) \mu_i +  \ve_2$ which implies that $ v_{z,i} \le \sqrt { \frac { 2 \mu_i + \ve_2 } { \lambda_z } }$. Moreover, since $\lambda_z \ge \max \{ (1-\ve_1)\mu_z -\ve_2 , 1\}$, we know that the projection of $v_z$ to $e_i$ will be smaller than $2 \sqrt {\frac { \mu_i +  \ve_2 } { \max \{ \mu_z - 2 \ve_2 , 1\} } }$.
An additional bound for the projection of $v_z$ to $e_i$ can be obtained by considering the variance of the matrices $A$ and $B$ in the direction $v_z$. Since we know that
$v^T_z B v_z \ge (1-\ve_1) v^T_z A v_z - \ve_2$, we get that $\sum_i \mu_i (v_z e_i)^2 \le \frac{1}{1-\ve_1}\left( \lambda_z +  \ve_2 \right) \le 2 ( \lambda_z +  \ve_2)$ which implies that $v_{z,i} \le \sqrt { 2 \frac { \lambda_z +  \ve_2 } { \mu_i } }$.

We now guess vectors $v'_1,...,v'_k$ that approximate the eigenvectors of $B$ by additively gridding over the projections to each eigenvector of $A$. To get an approximation $v'_z$ of the eigenvector $v_z$, we grid over a projection to $e_i$ with accuracy $\ve' \min \left\{2 \sqrt {\frac { \mu_i +  \ve_2 } { \max \{ \mu_z - 2 \ve_2 , 1\} } }, 1 \right\}$ for a small enough $\ve'$ that only depends on $k$, $\ve_1$ and $\ve_2$. This requires $\frac 1 {\ve'}$ guesses for each projection, and thus $\left( \frac 1 {\ve'} \right)^{k^2}$ guesses for all $k^2$ projections. The final covariance matrix we output is then $\hat B = \sum_z \lambda'_z v'_z (v_z')^T$.

We will now show that the covariance matrix $\hat B$ satisfies the property that it is close in all directions to $B$. To do this we will make use of Lemma~\ref{lem:1dtokdestimate}, and only consider directions $y = \frac {v_z} {\sqrt  \lambda_z}$ for $z \in [k]$ and $ y = \frac {v_z} {\sqrt  \lambda_z} + \frac {v_{z'}} {\sqrt \lambda_{z'}}$ for $z,z' \in [k]$.

We now consider direction $y = \frac {v_z} {\sqrt  \lambda_z}$. We have that:
$$\frac {v_z^T} {\sqrt  \lambda_z} \hat B \frac {v_z} {\sqrt  \lambda_z} = \sum_i \frac {\lambda'_i} {\lambda_z} (v_z v'_i)^2 = \sum_i \frac {\lambda'_i} {\lambda_z} (v_z v_i + v_z (v'_i - v_i) )^2 = \frac {\lambda'_z} {\lambda_z} (1 + v_z (v'_z -v_z) )^2 + \sum_{i \neq z} \frac {\lambda'_i} {\lambda_z} (v_z (v'_i -v_i) )^2$$

The first term is in the range $[(1-\ve)(1-k \ve')^2,(1+\ve)(1+k \ve')^2]$, which for $\ve' \le \ve/k$, becomes $(1 \pm O(\ve))$. The rest of the terms can be bounded as follows:

\begin{align*}
\frac {\lambda'_i} {\lambda_z} (v_z (v'_i -v_i) )^2 &\le (1+\ve) \frac {\lambda_i} {\lambda_z}  ( \sum_j v_{z,j} (v'_{i,j} -v_{i,j}) )^2 \\
&\le  (1+\ve) \frac {\lambda_i} {\lambda_z} \left( \sum_j \sqrt { 2 \frac {\lambda_z +  \ve_2 } { \mu_j } } \ve' 2 \sqrt {\frac { \mu_j +  \ve_2 } { \max \{ \mu_i - 2 \ve_2 , 1\} } } \right)^2 \\
&\le  (1+\ve) \frac {\lambda_i} {\lambda_z} \left( \sum_j 2 \ve' \sqrt{2 \lambda_z} \sqrt { \frac { 1 +  \ve_2 } {1} }  \sqrt { \frac { 1 +  \ve_2 } { \max \{ \mu_i - 2 \ve_2 , 1\} } } \right)^2 \\
&\le  (1+\ve) \left( \sum_j 2 ( 1 +  \ve_2 ) \ve'   \sqrt { \frac { 2 \lambda_i } { \max \{ \mu_i - 2 \ve_2 , 1\} } } \right)^2 \\
&\le  (1+\ve) \left( 4 k ( 1 +  \ve_2 ) \ve'   \sqrt { \frac { \mu_i + \ve_2 } { \max \{ \mu_i - 2 \ve_2 , 1\} } } \right)^2 \\
&\le  (1+\ve) \left( 8 k ( 1 +  \ve_2 ) \sqrt \ve_2  \ve' \right)^2 \le \frac {\ve} k
\end{align*}
for $\ve' = O( \sqrt \ve ((1 + \ve_2) k)^{-3/2} )$. This means that $v^T_z \hat B v_z \in (1 - \ve, 1 + \ve) \l_z$. The proof is similar for directions $y = \frac {v_z} {\sqrt  \lambda_z} + \frac {v_{z'}} {\sqrt \lambda_{z'}}$ for $z,z' \in [k]$.

Overall, we can get an estimate $\hat B$ of any matrix $B \in S$ by making at most $\left(\frac {k (1 + \ve_2)} { \ve} \right)^{O(k^2)}$ guesses, which implies an $\ve$-cover of this size.
\end{proof}

Applying Lemma~\ref{lem:psdcover} for $\ve_1 = 2\ve$ and $\ve_2 = m \le t k^2$, it is easy to see that we can get a good estimate $\hat \S_G$ of $\S_G$ using only $\left(\frac {k} {\ve} \right)^{O(k^2)}$ guesses.
This completes the analysis for obtaining an accurate estimate for the covariance matrix. 
The same approach also gives us an accurate estimate for the mean vector. 
We guess the projection of the mean on each of the (approximate) eigenvectors with accuracy proportional to the square root of the corresponding eigenvalue as in Lemma~\ref{lem:psdcover}. This requires only $\left(\frac {k} {\ve} \right)^{O(k)}$ additional guesses, so overall we can compute the estimates $(\hat \mu_G, \hat \S_G)$ using only $\left(\frac {k} {\ve} \right)^{O(k^2)}$ guesses.

\subsection{Probability Density Computation}
\label{sec:pmdpdf}
In order to apply Theorem~\ref{thm:tournament}, we need access to a PDF comparator (Definition~\ref{def:pdfcomp}).
We will implement this by explicitly computing the probability mass function (PMF) of a distribution at a given point $x$.
The naive computation could require time which is polynomial in $n$ or exponential in $1/\ve$.
We will show how to avoid these costs using a dynamic program.

\begin{lemma}
There exists an algorithm which computes the probability mass function for the convolution of a discretized Gaussian with a $(\poly(k/\ve),k)$-PMD at a given point $x$ in time $(k/\ve)^{O(k)}$.
\end{lemma}

\begin{proof}
Let $G(\cdot)$ and $G_d(\cdot)$ be the PDF and PMF of the non-discretized and the discretized Gaussian, respectively.
Similarly, let $PMD(\cdot)$ be the PMF of the $(\poly(k/\ve),k)$-PMD. For any given integer point $x$, we can compute $G_d(x)$ by computing the integral of the non-discretized Gaussian in a unit box around $x$, i.e. by letting $R(x) = \prod_i [x_i - 1/2, x_i + 1/2]$, we have that $G_d(x) = \int_{R(x)} G(t) dt$. We can compute this integral with very high accuracy using numerical integration methods.

To compute $PMD(x)$, we use dynamic programming.
We will maintain the variables $P(x,i)$ that give us the probability at the point $x$ in the support of the PMD considering only the first $i$ CRVs. 
It is easy to compute $P(x,i)$ as
$\sum_j \r_{i,j} P(x-e_j,i-1)$, where $\r_{i,j}$ is the probability the $i$-th CRV assigns to coordinate $j$. Since there are $(k/\ve)^{O(k)}$ points in the support of the PMD and at most $\poly(k/\ve)$ CRVs in the PMD, we can compute the probability density for the whole support of the PMD in time $(k/\ve)^{O(k)}$.

To compute the probability density at point $x$ for the convolution of $PMD(\cdot)$ with $G_d(\cdot)$, we write it as: $\sum_{y} PMD(y) G_d(x-y)$. We only need to consider the summation for points $y$ in the support of the $(\poly(k/\ve),k)$-PMD. 
Since there at most $(k/\ve)^{O(k)}$ such points the lemma follows.
\end{proof}

\section{Details from Section \ref{sec:learningsiirv}}
\label{sec:appendixsiirv}

We first recall the main structural result from \cite{DaskalakisDOST13}:
\siirvstruct*

Now, we provide learning algorithms for the two cases, corresponding to Lemmas 5.1 and 5.2 in \cite{DaskalakisDOST13}.
\begin{lemma}
  \label{lem:siirvsparse}
  There is a procedure \texttt{Learn-Sparse} with the following properties: It takes as input an accuracy parameter $\ve' > 0$, and a confidence parameter $\d' > 0$, as well as access to samples from an $(n,k)$-SIIRV $S$. 
  \texttt{Learn-Sparse} uses $m = k^{5k} \cdot \tilde O(\log^{k+2}(1/\ve) \log 1/\d' / \ve'^2)$ samples from $S$ and has the following guarantee:
  If the variance of $S$ is at most $15(k^{18}/\ve'^6)\log^2(1/\ve')$, then we return a hypothesis variable $H_c$ such that $\dtv(S, H_c) = O(\ve')$ with probability at least $1 - \d'$. 
\end{lemma}
\begin{proof}
  Let $S = \sum_{i = 1}^n S_i$, and $S^{PMD} = \sum_{i = 1}^n S^{PMD}_i$ be a $(n,k)$-PMD such that $(1, \dots, k)^T \cdot S^{PMD}_i = S_i$ for all $i$.
  We will apply the rounding procedure described in Section \ref{sec:rounding} on $S^{PMD}$ to argue that $S$ is close to a shifted $(\poly(k/\ve'), k)$-SIIRV.
  This will be sufficient to complete the proof, as we can $\ve'$-cover $S$ by considering a cover of all unshifted $(\poly(k/\ve'),k)$-PMDs when projected onto $(1, \dots, k)$, which, by Theorem \ref{thm: cover}, is a set of size $N = 2^{O(k^2 \log (k/\ve') + (k^{5k} \log^{k+2} (1/\ve'))}$.
    The shift is determined by taking $O(1/\ve'^2)$ samples and trying all integers within an additive $\poly(k/\ve')$ of the mean of these samples.
    We select one of these hypotheses using Theorem \ref{thm:tournament}, requiring $\frac{1}{\ve'^2} \log N = k^{5k} \cdot \tilde O(\log^{k+2}(1/\ve) \log 1/\d'/\ve'^2)$ samples, as desired.
    
    Let $T^{PMD}$ be result when the rounding procedure of Section \ref{sec:rounding} is applied to $S^{PMD}$, and let $T = \sum_{i=1}^n T_i$ where $T_i =  (1, \dots, k)^T \cdot T^{PMD}_i$.
    By Lemma \ref{lem:round} and the Data Processing Inequality (Lemma \ref{lem:DPI}), this tells us that
    $$\dtv\left(S, T\right) \leq \dtv\left(S^{PMD}, T^{PMD}\right) \leq O\left(c^{1/2} k^{5/2} \log^{1/2}\left(\frac{1}{ck}\right)\right) = O(\ve'),$$
    where the last equality follows by our choice of $c$.
    We prove by contradiction that $T$'s variance is still $\poly(k/\ve')$.
    Suppose not, that $\s_T^2 \geq \z^2$, where $\z^2 \coloneqq \poly(k/\ve')$.
    We apply the Berry-Esseen theorem (Proposition \ref{prop:berryesseen}) to $T'$, which is a re-centered version of $T$.
    Defining $T_i' \sim T_i - E[T_i]$, and observing that $T_i' \in [-k, k]$ we note that $\m_{T'} = 0, \s^2_{T'} = \s^2_T, \r_{T'} \leq k \s^2_T$.
    Thus,
    $$\dk(T, \mathcal{N}(\m_T, \s^2_T)) \leq \frac{k}{\z}.$$
    By triangle inequality and the fact that total variation distance upper bounds Kolmogorov distance,
    $$\dtv(S, \mathcal{N}(\m_T, \s^2_T)) \leq O(\ve') + \frac{k}{\z}.$$
    However, anticoncentration of a Gaussian tells us that for any point $x$,
    $$\Pr(|\mathcal{N}(\m_T, \s^2_T) - x| \leq \ell) \leq \sqrt{\frac{2}{\p}} \frac{\ell}{\z}.$$
    Examine the interval of width $k^9/2\ve'^4$ centered at $E[S]$.
    $S$ assigns at least $1 - \ve'$ mass to this interval, but $\mathcal{N}(\m_T, \s^2_T)$ assigns at most $\sqrt{\frac{2}{\p}}\frac{k^9}{2\ve'^4}/\z$ mass.
    If $|(1 - \ve') - \sqrt{\frac{2}{\p}}\frac{k^9}{2\ve'^4}/\z| > O(\ve') + \frac{k}{\z}$, which happens for $\z = \omega(k^9/\ve'^4)$, this interval demonstrates that the total variation distance is larger than we showed above, thus arriving at a contradiction.
    Thus, we have that the variance of $T$ is at most $\zeta^2$.
    
    By the rounding procedure, we know that the variance of any $T_i$ which is non-constant is at least $c(1 - c) \geq c/2$.
    Since variance is additive and the variance of $T$ is at most $\z^2$, this implies that there are at most $2\z^2/c = O(k^{24}/\ve'^{11})$ non-constant $T_i$.
    Therefore, $S$ is $\ve'$-close to a shifted $(O(k^{24}/\ve'^{11}),k)$-SIIRV, as desired.
\end{proof}

\begin{lemma}
  \label{lem:siirvheavy}
  There is a procedure \texttt{Learn-Heavy} with the following properties:
  It takes as input a value $\ell \in \{1, \dots, k-1\}$, an accuracy parameter $\ve' > 0$, a variance parameter $\s^2 = \Omega(k^2/\ve'^2)$, and a confidence parameter $\d' > 0$, as well as access to samples from a $\poly(n)$-IRV $S$.
  \texttt{Learn-Heavy} uses $m = O((1/\ve'^2)(\ell+ \log(1/\d')))$ samples from $S$, runs in time $\tilde O(m)$, and has the following performance guarantee:
  
  Suppose that $\dtv(S, \ell Z + Y) \leq \ve'$, where $Z$ is a discretized random variable distributed as $\lfloor \mathcal{N}\left(\frac{\m'}{\ell}, \frac{\s'^2}{\ell^2}\right) \rceil$, for some $\s'^2 \geq \s^2$, $Y$ is a $\ell$-IRV, and $Z$ and $Y$ are independent.
  Then \texttt{Learn-Heavy} outputs a hypothesis variable $H_\ell$ such that $\dtv(S,H_\ell) \leq O(\ve')$ with probability at least $1 - \d'$.
\end{lemma}
\begin{proof}
  This follows similarly to the proof of Lemma 5.2 in \cite{DaskalakisDOST13}. 
  $Y$ and $Z$ are learned in separate stages.
  $Y$ is learned identically as in their algorithm, by using the empirical distribution of $O((1/\ve'^2)(\ell + \log(1/\d')))$ samples reduced to their residue mod $\ell$.

  Learning $Z$ is performed differently.
  We take $O(\log(1/\d')/\ve'^2)$ samples and replace each value $v$ with the value $\lfloor v/\ell \rfloor$.
  In other words, given samples from $S$, we simulate samples from $Z' = \lfloor S/\ell \rfloor$.
  Since $\dtv(S, \ell Z + Y) \leq \ve'$, Lemma \ref{lem:DPI} implies that $\dtv(Z', Z) \leq \ve'$, which in turn implies that $\dk(Z', Z) \leq \ve'$, using Fact \ref{fct:drelation}.

  Using Lemma \ref{lem:dkw}, our samples from $S$ give us a distribution $\hat Z'$ such that $\dk(Z', \hat Z') \leq \ve'$ with probability $1 - \d'$. 

We make the following straightforward observation, bounding the Kolmogorov distance between a Gaussian and the corresponding discretized Gaussian.
\begin{proposition}
  $\dk(\mathcal{N}(\m,\s^2), \lfloor \mathcal{N}(\m,\s^2)\rceil) \leq \frac{1}{\s \sqrt{2\p}}.$
\end{proposition}

Using triangle inequality and the lower bound on $\s^2$, this tells us that $\dk(\hat Z', \mathcal{N}(\frac{\m'}{\ell},\frac{\s'^2}{\ell^2})) \leq O(\ve')$. 

Now, we can apply the following robust statistics results from \cite{DaskalakisK14}:
\begin{lemma}[Lemmas 9 and 10 of \cite{DaskalakisK14}]
  Let $\hat F$ be a distribution such that $\dk(\mathcal{N}(\m,\s^2), \hat F) \leq \ve$.
  Then
  \begin{itemize}
    \item $med(\hat F) \triangleq \hat F^{-1}(\frac12) \in \mu \pm O(\ve \s)$ 
    \item $\frac{IQR(\hat F)}{2\sqrt{2}erf^{-1}(\frac12)} \triangleq \frac{\hat F^{-1}(\frac34) - \hat F^{-1}(\frac14)}{2\sqrt{2}erf^{-1}(\frac12)} \in \s \pm O(\ve \s)$
  \end{itemize}
\end{lemma}

By taking the median and a rescaling of the interquartile range of $\hat Z$, we get estimates $\hat \m'$ and $\hat \s'^2$ which are within $\pm O(\ve \frac{\s'}{\ell})$ of the true parameters.
Proposition \ref{prop:ddodtv} implies $\dtv(\mathcal{N}(\hat \m',\hat \s'^2), \mathcal{N}(\frac{\m'}{c},\frac{\s'^2}{c^2})) \leq O(\ve')$. 
Applying Lemma \ref{lem:DPI} gives us 
$\dtv(\lfloor \mathcal{N}(\hat \m',\hat \s'^2) \rceil, Z) \leq O(\ve')$.
The result follows using triangle inequality on the estimates for $Y$ and $Z$.
\end{proof}

Now, we run \texttt{Learn-Sparse} once, and \texttt{Learn-Heavy} for $c = 1$ to $k-1$. 
This will give us a set of $k$ hypotheses, at least one of which is close to the true distribution.
We use the subroutine \texttt{FastTournament} (as described by Theorem \ref{thm:tournament}) to select one of these hypotheses.
Theorem \ref{thm:SIIRV learning} follows by combining Lemma \ref{lem:siirvstruct} with the guarantees provided by Lemmas \ref{lem:siirvsparse} and \ref{lem:siirvheavy}.

\end{document}